\documentclass[%
reprint,
nofootinbib,
superscriptaddress,
amsmath,
amssymb,
prb,
nobibnotes
]{revtex4-1}
\usepackage{braket}
\usepackage{commath}
\usepackage{etoolbox}

\usepackage[usenames,dvipsnames]{xcolor}

\usepackage{makecell}
\usepackage{multirow}
\usepackage{cellspace}

\usepackage{graphicx}
\usepackage{bm}
\usepackage{array}
\usepackage{enumitem}

\usepackage{todonotes}
\usepackage{soul}

\usepackage{bbold}

\usepackage{balance}

\usepackage{changes}
\colorlet{Changes@Color}{magenta}
\newcommand{\stkout}[1]{\ifmmode\text{\sout{\ensuremath{#1}}}\else\sout{#1}\fi}
\setdeletedmarkup{\stkout{#1}}

\usepackage[super]{nth}
\newcommand{\RNum}[1]{\uppercase\expandafter{\romannumeral #1\relax}}

\usepackage{todonotes}
\usepackage{soul}

\usepackage{bbold}
\usepackage{xcolor}

\allowdisplaybreaks

\usepackage{ifxetex,ifluatex}
\ifnum 0\ifxetex 1\fi\ifluatex 1\fi=0 
  \usepackage[T1]{fontenc}
  \usepackage[utf8]{inputenc}
\else 
  \ifxetex
    \usepackage{mathspec}
  \else
    \usepackage{fontspec}
  \fi
  \defaultfontfeatures{Ligatures=TeX,Scale=MatchLowercase}
\fi
\IfFileExists{upquote.sty}{\usepackage{upquote}}{}
\IfFileExists{microtype.sty}{%
\usepackage[]{microtype}
\UseMicrotypeSet[protrusion]{basicmath} 
}{}
\PassOptionsToPackage{hyphens}{url} 
\usepackage[unicode=true,colorlinks=true,citecolor=blue,urlcolor=blue]{hyperref}
\hypersetup{
            pdfborder={0 0 0},
            breaklinks=true}
\usepackage{graphicx,grffile}
\makeatletter
\def\maxwidth{\ifdim\Gin@nat@width>\linewidth\linewidth\else\Gin@nat@width\fi}
\def\maxheight{\ifdim\Gin@nat@height>\textheight\textheight\else\Gin@nat@height\fi}
\makeatother
\setkeys{Gin}{width=\maxwidth,height=\maxheight,keepaspectratio}
\IfFileExists{parskip.sty}{%
\usepackage{parskip}
}{
\setlength{\parindent}{0pt}
\setlength{\parskip}{6pt plus 2pt minus 1pt}
}
\ifx\subsection\undefined\else
\let\oldsubsection\subsection
\renewcommand{\subsection}[1]{\oldsubsection{#1}\mbox{}}
\fi
\ifx\subsubsection\undefined\else
\let\oldsubsubsection\subsubsection
\renewcommand{\subsubsection}[1]{\oldsubsubsection{#1}\mbox{}}
\fi

\makeatletter
\def\fps@figure{htbp}
\makeatother

\begin{document}
\title{Hidden-symmetry-enforced nexus points of nodal lines in layer-stacked dielectric photonic crystals}

\author{Zhongfei Xiong}\thanks{These authors contributed equally to this work.}%
\affiliation{School of Optical and Electronic Information, Huazhong University of Science and Technology, Wuhan 430074, China}

\author{Ruo-Yang Zhang}\thanks{These authors contributed equally to this work.}
\affiliation{Department of Physics, The Hong Kong University of Science and Technology, Clear Water Bay, Hong Kong, China}

\author{Rui Yu}
\affiliation{School of Physics and Technology, Wuhan University, Wuhan 430072, China}

\author{C. T. Chan}
\email{phchan@ust.hk}
\affiliation{Department of Physics, The Hong Kong University of Science and Technology, Clear Water Bay, Hong Kong, China}

\author{Yuntian Chen}
\email{yuntian@hust.edu.cn}
\affiliation{School of Optical and Electronic Information, Huazhong University of Science and Technology, Wuhan 430074, China}
\affiliation{Wuhan National Laboratory of Optoelectronics, Huazhong University of Science and Technology, Wuhan 430074, China}
\begin{abstract}
It was recently demonstrated  that the connectivities of bands emerging from zero frequency in dielectric photonic crystals are distinct from their electronic counterparts with the same space groups. We discover that, in an AB-layer-stacked photonic crystal composed of anisotropic dielectrics, the unique photonic band connectivity leads to a new kind of symmetry-enforced triply degenerate points at the nexuses of two nodal rings and a Kramers-like nodal line. The emergence and intersection of the line nodes are guaranteed by a generalized 1/4-period screw rotation symmetry of Maxwell's equations. The bands with a constant $k_z$ and iso-frequency surfaces near a nexus point both disperse as a spin-1 Dirac-like cone, giving rise to exotic transport features of light at the nexus point. We show that the spin-1 conical diffraction occurs at the nexus point which can be used to manipulate the charges of optical vortices. Our work reveals that Maxwell's equations can have hidden symmetries induced by the fractional periodicity of the material tensor components and hence paves the way to finding novel topological nodal structures unique to photonic systems.
\end{abstract}

\maketitle
\textbf{Introduction}\\[3pt]
Discovering and synthesizing symmetry-protected topological (SPT) band degeneracies, including nodal points \cite{wan2011Topological,liu2014Discoverya, xu2015Discovery,armitage2018Weyl,lu2013Weyl,lu2015Experimental,chen2016Photonic,noh2017Experimental,chang2017Multiplea,yang2018Ideal,wang2016Threedimensional,wang2017TypeII,guo2019Observation} and nodal lines (NLs) \cite{weng2015Topological,chan2016MathrmCa,fang2016Topological,kim2015Dirac,yu2015Topological,gao2018Class,zhang2018Hybrid,he2018TypeII,kawakami2017Symmetryguaranteed,yan2018Experimental,gao2018Experimental,xia2019Observation}, is a rapidly growing frontier in the field of topological materials.  The initial impetus for the area came from realizing elusive relativistic quasi-particles, \textit{e.g.} 3-dimensional (3D) Weyl and Dirac fermions, in both electronic crystalline materials~\cite{wan2011Topological,liu2014Discoverya, xu2015Discovery,armitage2018Weyl} and photonic crystals (PhCs)~\cite{lu2013Weyl,lu2015Experimental,chen2016Photonic,noh2017Experimental,chang2017Multiplea,yang2018Ideal,wang2016Threedimensional,wang2017TypeII,guo2019Observation}. Interestingly, since the crystallographic space groups impose fewer constraints on the energy bands than the continuous Poincar\'e group, more exotic multifold band crossings were found in lattice systems~\cite{bradlyn2016Dirac}, which have no counterparts in high-energy physics.   
As a remarkable example, certain space groups allow the existence of triply degenerate points in the band structures, forming either as isolated point nodes carrying monopole charges, so-called spin-1 Weyl points~\cite{bradlyn2016Dirac,saba2017Group,hu2018Topological,yang2019Topological,zhang2018DoubleWeyl}, or as nexuses connecting several NLs~\cite{zhu2016Triple,chang2017Nexus,lv2017Observation,chang2017Nexus}. On the other hand, the SPT band crossings can also be classified according to whether they are merely symmetry-allowed (accidental) or symmetry-enforced~\cite{zhang2018Topological,chan2019Symmetryenforced}. The former are only perturbatively stable, whereas the symmetry-enforced degeneracies 
are robust against any large symmetry-preserving deformations and are currently drawing more attention  due to their deterministic nature~\cite{zhang2018Topological,chan2019Symmetryenforced,xia2019Observation}.

In PhCs, the topology of band structures is usually thought to be adequately described by spinless space groups, provided that special internal symmetries, such as electromagnetic (EM) duality, are not imposed on the EM materials. However, in dielectric PhCs, there are always two gapless bands emerging from zero frequency and momentum, $\omega=|\mathbf{k}|=0$, irrespective of the space group representations at that point. Watanabe and Lu recently revealed  that this intrinsic singularity of EM fields permits higher minimal connectivity for the lowest photonic bands than for their electronic counterparts without spin-orbit coupling, and may further enforce unique photonic band crossings even in symmorphic lattices~\cite{watanabe2018Space}. This pioneering study uncovered the tip of the hidden characteristics of Maxwell's equations that are relevant to photonic band connectivities. In general, the stationary Maxwell's equations can be written as a generalized eigenvalue problem,
\begin{equation}\label{maxwell}
   \begin{pmatrix}
    0 & i\nabla\times\\
    -i\nabla\times & 0
    \end{pmatrix}
    \begin{pmatrix}
    \mathbf{E}\\ \mathbf{H}
    \end{pmatrix}
    =\omega \begin{pmatrix}
    \tensor{\varepsilon}(\mathbf{r}) & \tensor{\chi}(\mathbf{r})\\
    \tensor{\chi}(\mathbf{r})^\dagger & \tensor{\mu}(\mathbf{r})
    \end{pmatrix}
    \begin{pmatrix}
    \mathbf{E}\\ \mathbf{H}
    \end{pmatrix},
\end{equation}
where we henceforth denote the curl matrix and the constitutive matrix  on the left and right sides of Eq.~\eqref{maxwell} as $\hat{\mathcal{N}}$ and $\hat{\mathcal{M}}(\mathbf{r})$, respectively.
Since all space group transformations leave the curl matrix $\hat{\mathcal{N}}$ invariant, a PhC respects a space group symmetry $\hat{A}$, only if its constitutive tensor obeys $\hat{A}\,\hat{\mathcal{M}}(\mathbf{r})\hat{A}^{-1}=\hat{\mathcal{M}}(\mathbf{r})$. However, a generic symmetry $\widetilde{A}$ of Maxwell's equations~\eqref{maxwell} operates on the Hamiltonian $\hat{H}(\mathbf{r})=\hat{\mathcal{M}}(\mathbf{r})^{-1}\hat{\mathcal{N}}$ of EM fields, namely requiring $\widetilde{A}\hat{H}(\mathbf{r})\widetilde{A}^{-1}=\hat{H}(\mathbf{r})$, and not on $\hat{\mathcal{N}}$ and $\hat{\mathcal{M}}(\mathbf{r})$ separately. This fact implies that the conventional space groups alone are insufficient to determine the symmetry properties as well as the band connectivities of photonic systems.

\begin{figure*}[hbt]
  \includegraphics[width=0.95\textwidth]{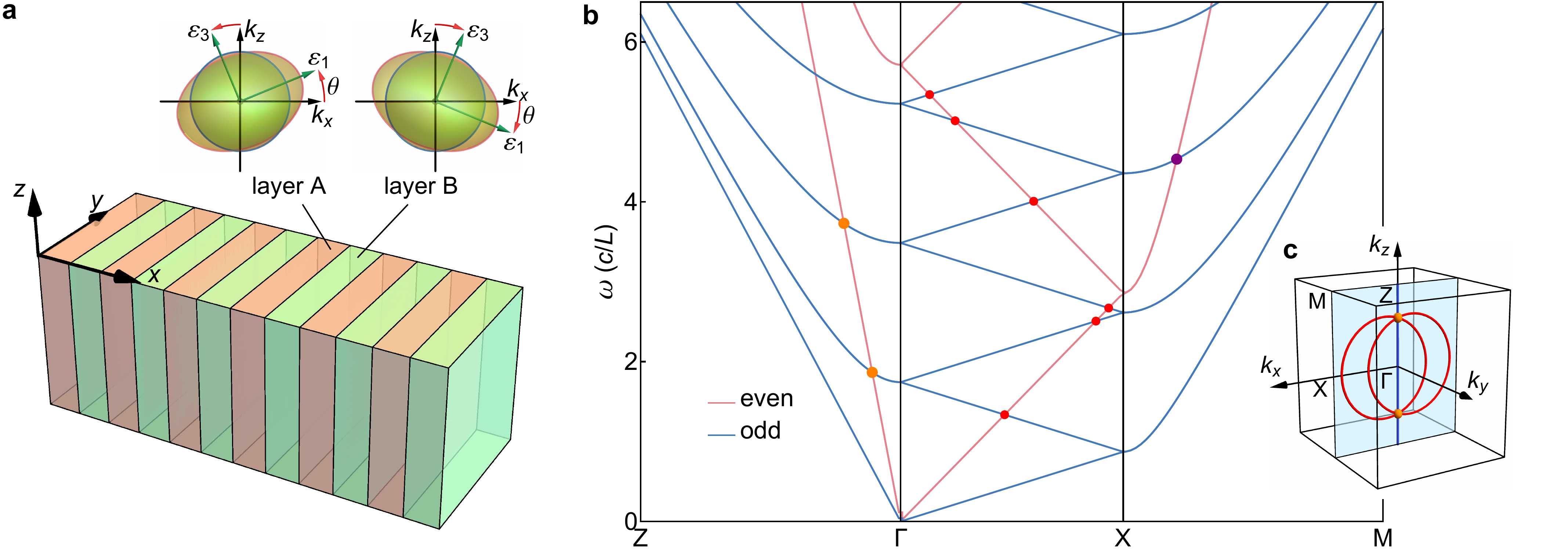}
  \caption{\label{fig1}AB-layer-stacked photonic crystal made of a generic biaxial dielectric. \textbf{a} Structure of the PhC, where the two insets display the iso-frequency surfaces in the $xz$ plane of the biaxial dielectrics in layer A (orange) and layer B (green). \textbf{b} Band structure along high symmetry lines in the $k_y=0$ plane for the PhC with $\varepsilon_1=2$, $\varepsilon_2=13$, $\varepsilon_3=1$, and rotation angle $\theta=\pi/5$. The blue and magenta lines represent the bands with odd and even $\hat{M}_y$ parities respectively. The red dots correspond to the nodal lines along which two bands with opposite $\hat{M}_y$ parities intersect. The orange and purple dots denote the threefold and fourfold degenerate nexus points respectively.  \textbf{c} Sketch map of two red nodal rings (corresponding to the two lowest red dots in \textbf{b}) crossing a Kramers-like nodal line (blue line) at two nexus points in the first (bulk) Brillouin zone of the PhC.}
  \end{figure*}

In this work, we propose a simple layer-stacked photonic structure consisting of anisotropic dielectrics to exemplify such hidden symmetries of Maxwell's equations beyond space groups. We show that a hidden symmetry, more specifically, a generalized fractional screw rotation symmetry, together with time reversal symmetry guarantees the emergence of Kramers-like straight NLs passing through the Brillouin zone (BZ) center, and results in unusual photonic band connectivities. Furthermore, we demonstrate that the lowest Kramers-like NL can almost always intersect with two other SPT nodal rings at two triply degenerate nexus points (NPs), which can be seen as a new kind of magnetic monopole connecting Berry flux strings in  momentum space~\cite{cornwall1999Center,volovik2000Monopoles,heikkila2015Nexus}. By breaking the hidden symmetry, we lift the two NPs and achieve type-II and type-III nodal rings in the PhC~\cite{gao2018Class,zhang2018Hybrid,he2018TypeII}. 
In addition, the peculiar anisotropic band structure in the vicinity of the NPs, especially the spin-1 conical dispersion of the iso-frequency surfaces, may lead to novel transport phenomena. As an example, we show how  optical vortices can be manipulated via the spin-1 conical diffraction effect~\cite{diebel2016Conical} for a beam incident at an NP.

\textbf{Results}\\[3pt]
The photonic crystal considered here consists of two types of dielectric layer (A and B) stacked periodically along the $x$ direction which are homogeneous in the transverse $yz$ plane. The A and B layers have equal thicknesses, $L/2$, and are both composed of the same kind of nondispersive anisotropic dielectric with principal relative permittivity values of $\varepsilon_1,\ \varepsilon_2,\ \varepsilon_3\,(\neq\varepsilon_1)$, whereas the optical axes of the dielectric rotate alternatively in the AB layers as shown in Fig.~\ref{fig1}a. Specifically, the second principal axis of the materials is fixed along the $y$ direction, while the first principle axis is rotated by an angle $\theta$ counterclockwise (clockwise) from the $x$ axis in layer A (B). As such, the PhC's relative permittivity tensor in $xyz$ coordinates is given by 
\begin{equation}\label{epsilon}
\tensor{\bm \varepsilon}_r= \begin{pmatrix}
\varepsilon_{xx}&0&\varepsilon_{xz}\\
0&\varepsilon_{yy}&0 \\
\varepsilon_{zx}&0&\varepsilon_{zz}\\
\end{pmatrix},
\end{equation}
where $\varepsilon_{xz}=\varepsilon_{zx}=\pm g=\pm\left(\varepsilon_1-\varepsilon_3 \right)\sin{\theta}\cos{\theta}$ flips its sign between layers A and B, while the diagonal elements $\varepsilon_{xx}=\left(\cos^2{\theta}\varepsilon_1+\varepsilon_3 \sin^2{\theta}\right)$, $\varepsilon_{yy}=\varepsilon_2$, and $\varepsilon_{zz}=\left(\sin^2{\theta}\varepsilon_1+\varepsilon_3 \cos^2{\theta}\right)$ are all constant. The band structure along high symmetry lines of the PhC is plotted in Fig.~\ref{fig1}b (see supplementary information S1 for the analytical calculation).

\textbf{\textit{Space group symmetries}}\\
The space group of the PhC is $\mathbb{R}^2\rtimes\mathrm{Rod}(22)$, \textit{i.e.} the semidirect product of the 2-dimensional continuous translational group $\mathbb{R}^2$  in the $yz$ plane and the nonsymmorphic rod group \textbf{22} ($pmcm$)~\cite{kopsky2002international} associated with discrete translations along the $x$-axis. 
Here, we focus on several space group symmetries relevant to the band crossings in the $k_y=0$ plane.

First, the PhC is invariant under the mirror reflection ($\hat{M}_y$) about the $y=\mathrm{constant}$ planes, which permits the bands with opposite (odd and even) mirror parities to intersect along NLs in the $\hat{M}_y$-invariant plane $k_y=0$~\cite{weng2015Topological,chan2016MathrmCa,fang2016Topological}, as marked by the red dots in Fig.~\ref{fig1}b and the red rings in Fig.~\ref{fig1}c (also see Fig.~\ref{fig2} for the 3D band structures).
Second, the combined inversion and time reversal symmetry ($\mathcal{PT}$) quantizes the Berry phase encircling the nodal lines as $\pi$, stabilizing the nodal lines against local $\mathcal{PT}$-preserving perturbations~\cite{kim2015Dirac,yu2015Topological,chan2016MathrmCa,fang2016Topological}. 
Third, the PhC has a  twofold screw symmetry $\hat{S}_{2x}:(x,y,z)\to (x+\frac{L}{2},-y,-z)$. Together with $\mathcal{T}$, the combined symmetry $\hat{\Theta}_{L/2}=\mathcal{T}\hat{S}_{2x}$ ensures that all Bloch states are doubly degenerate in the $k_x=\pm\pi/L$ plane (corresponding to the twofold degenerate bands along $X-M$ in Fig.~\ref{fig1}b)~\cite{wang2017TypeII}.

However, although the space group only supports 1D irreducible representations along the $\Gamma-Z$ line, the band structure shows that two bands with the same $\hat{M}_y$-parity always linearly cross along this line regardless of the dielectric parameters, and accordingly the two red nodal rings intersect at two NPs (orange dots in Fig.~\ref{fig1}c). This indicates that the PhC system possesses a symmetry beyond the crystallographic space group.

\textbf{\textit{Hidden symmetry and Kramers-like nodal lines}}\\
Since the permittivity $\tensor{\varepsilon}(\mathbf{r})$ and the Hamiltonian $\hat{H}(\mathbf{r})=\hat{\mathcal{M}}(\mathbf{r})^{-1}\hat{\mathcal{N}}$ of EM fields are generically tensors, the periodicity of the system restricts the period of each component of $\tensor{\varepsilon}(\mathbf{r})$ to a fraction $1/n$ of the full period.  As mentioned in the introduction, the space group symmetries of the PhC, \textit{e.g.}, $\hat{A}$, are entirely encoded in the space-dependent constitutive tensor as $\hat{A}\,\hat{\mathcal{M}}(\mathbf{r})\hat{A}^{-1}=\hat{\mathcal{M}}(\mathbf{r})$. However, 
a generic symmetry $\widetilde{A}$ of Maxwell's equations~\eqref{maxwell} implies that the whole Hamiltonian is invariant under $\widetilde{A}$,  \textit{i.e.}, $\widetilde{A}\hat{H}(\mathbf{r})\widetilde{A}^{-1}=\hat{H}(\mathbf{r})$, but allows $\widetilde{A}\hat{\mathcal{M}}(\mathbf{r})\widetilde{A}^{-1}\neq \hat{\mathcal{M}}(\mathbf{r})$. Here, we show that such hidden symmetry can emerge from the fractional periodicity of the dielectric components in Eq.~\eqref{epsilon}, thereby giving rise to Kramers-like NLs along $\Gamma-Z$.

As an accessible entry point, we first consider the $\hat{M}_y$-odd subsystem in the $k_y=0$ plane. 
Since the PhC is homogeneous in the $y$ direction, the $\hat{M}_y$-odd eigenstates, $\psi^\mathrm{odd}=(E_y,H_x,H_z)^\intercal$, only depend  on $\varepsilon_{yy}$. As $\varepsilon_{yy}$ is a global constant for the PhC in Fig.~\ref{fig1}, the dispersion along $\Gamma-X$ can be regarded as simple folding of a linear band, giving rise to degeneracies at $\Gamma$ and $X$. Let us consider a relaxed condition $\varepsilon_{yy}(x+L/4)=\varepsilon_{yy}(x)$. In this case, the odd subsystem has a fractional period $L/4$, hence the width of the BZ of the subsystem is quadrupled in the $x$ direction. The $\hat{M}_y$-odd band structure in the original BZ can be obtained by folding the bands in the quadruple BZ twice. 
In the quadruple BZ, the time reversal symmetry ensures that the eigenstates at $\pm k_x$ have degenerate eigenfrequencies $\omega(k_x)=\omega(-k_x)$. After band folding, every pair of eigenstates with identical frequencies at $k_x=\pm2\pi/L$ is shifted onto the same point along the central line $k_x=0$. Consequently, the $(4m+2)^\mathrm{th}$ and $(4m+3)^\mathrm{th}$ $\hat{M}_y$-odd bands ($m\geq0$ is an integer) are degenerate along $\Gamma-Z$, as shown in Fig.~\ref{fig1}b.

\begin{figure*}[tbh]
  \includegraphics[width=0.9\textwidth]{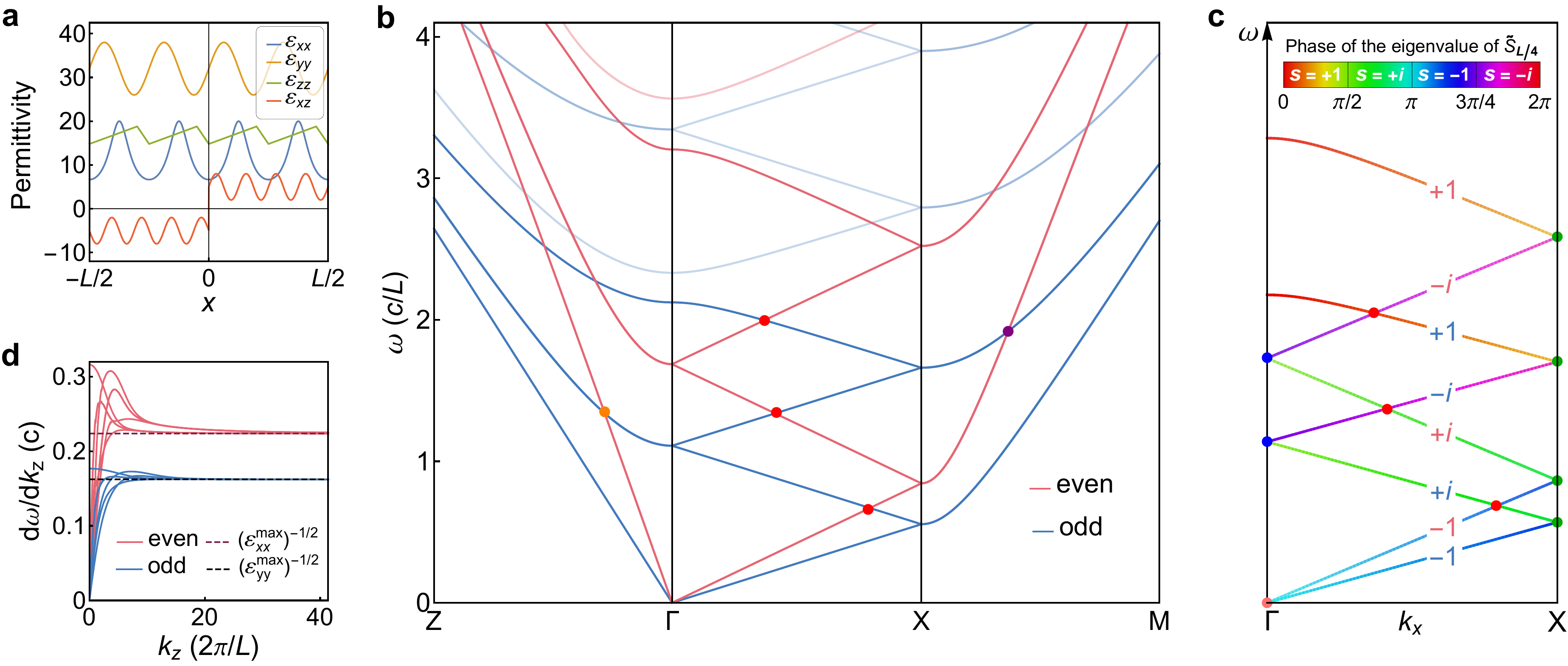}\caption{\label{connect} Consequence of the hidden symmetry $\widetilde{S}_{L/4}$ for band connectivity. \textbf{a} Profiles of the nonzero permittivity components in one period for a PhC obeying Eqs.~\eqref{condition1} and \eqref{condition2}. \textbf{b} Band structure along high symmetry lines for the PhC with the permittivity given by \textbf{a}. \textbf{c} Connectivity of the lowest four $\hat{M}_y$-even and lowest four $\hat{M}_y$-odd bands along $\Gamma-X$. The color at a point on the bands displays the phase of the eigenvalue of $\widetilde{S}_{L/4}$ at that point. The labels, $\pm1$, $\pm i$, denote the branch indices, $s$, of the bands.  \textbf{d} Group velocities $d\omega/d k_z$ of the even (magenta) and odd (blue) bands changing along $\Gamma-Z$, which converge to the asymptotic values $c/\sqrt{\varepsilon^\mathrm{max}_{xx}}$ and $c/\sqrt{\varepsilon^\mathrm{max}_{xx}}$ (two dashed lines), respectively,  as $k_z\to \infty$. }
\end{figure*}

Even though the $\hat{M}_y$-even subsystem in the $k_y=0$ plane, characterized by the submatrix $\begin{pmatrix}\varepsilon_{xx} &\varepsilon_{xz}\\\varepsilon_{zx}&\varepsilon_{zz}\end{pmatrix}$ of $\tensor{\varepsilon}_r$, has the same primitive period $L$ as the whole system, it can be demonstrated that the Hamiltonian of the subsystem with respect to the eigenvector ${\psi}'^{\mathrm{even}}=(D_x,E_z,H_y)^\intercal$  will only depend explicitly on $\varepsilon_{xx}$, $\varepsilon_{zz}$, and $\varepsilon_{xz}^2$ after a local $U(1)$ gauge transformation $\hat{U}(x,k_z)=\exp\left[ik_z\int_0^x\frac{\varepsilon_{xz}(\xi)}{\varepsilon_{xx}(\xi)}d\xi\right]$ (see supplementary information S2). For the AB-layer-stacked PhC in Fig.~\ref{fig1}, $\varepsilon_{xx}$, $\varepsilon_{zz}$, and $\varepsilon_{xz}^2$ are all constant, therefore the intersections of $\hat{M}_y$-even bands along $\Gamma-Z$ also result from the folding of a linear band. If we relax the constraint on the three parameters from being homogeneous to having a fractional period $L/n$, then band crossings along $\Gamma-Z$ can still exist. In fact, 4 is the minimum value of $n$ that maintains the space group $\mathbb{R}^2\rtimes\mathrm{Rod}(22)$ of the layer-stacked PhC and preserves the appearance of the NLs along $\Gamma-Z$. More specifically, the elements of the permittivity tensor should satisfy
\begin{gather}
    \varepsilon_{ii}(x+L/4)=\varepsilon_{ii}(x)\quad (i=x,y,z),\label{condition1}\\
    \hspace{-5pt}\varepsilon_{xz}(x+L/4)^2=\varepsilon_{xz}(x)^2\ \text{and}\ \varepsilon_{xz}(x+L/2)=-\varepsilon_{xz}(x),\label{condition2}
\end{gather}
where the second requirement in Eq.~\eqref{condition2} is necessary to keep the primitive period of the PhC of $L$.

If we focus on the minimum requirement case of $L/4$, 
we can introduce a generalized 1/4-period twofold screw operator, and prove that the complete Hamiltonian of both $\hat{M}_y$-even and odd subsystems in the $k_y=0$ plane is invariant under the generalized $1/4$-period twofold screw rotation (see supplementary information S2), 
\begin{equation}
   \widetilde{S}_{L/4}\hat{H}(k_y=0)\widetilde{S}_{L/4}^{-1}=\hat{H}(k_y=0).
\end{equation}
The generalized $1/4$-period twofold screw rotation about the $x$-axis is defined as
\begin{equation}
    \widetilde{S}_{L/4}=\big(\hat{P}_-+\hat{G}^{-1}\hat{U}^\dagger\hat{P}_+\big)
    \hat{C}_{2x}\hat{T}_x({\textstyle\frac{L}{4}})
    \big(\hat{P}_-+\hat{G}\hat{U}\hat{P}_+\big),
\end{equation}
where $\hat{C}_{2x}$ is the twofold rotation about the $x$-axis, $\hat{T}_x(\frac{L}{4})$ denotes the $1/4$-period translation along $x$ direction, $\hat{P}_\pm=\frac{1}{2}(\hat{I}_{6\times6}\pm\hat{M}_y)$ are the projection operators onto $\hat{M}_y$-even/odd subsystems, $\hat{U}$ is the aforementioned local $U(1)$ gauge transformation, and 
\begin{equation}
    \hat{G}=\hat{I}_{6\times6}+(\varepsilon_{xx}(x)-1)\hat{\mathbf{e}}_1\hat{\mathbf{e}}_1+\varepsilon_{xz}(x)\hat{\mathbf{e}}_1\hat{\mathbf{e}}_3
\end{equation}
transforms the eigenvector from $\Psi=(\mathbf{E},\mathbf{H})^\intercal$ to ${\Psi}'=(D_x,E_y,E_z,\mathbf{H})^\intercal$ with the basis $(\hat{\mathbf{e}}_i)_j=\delta_{ij}\ (i,j=1,\cdots,6)$.  In Fig.~\ref{connect}, we illustrate a more general example of PhCs satisfying Eqs.~\eqref{condition1} and (\ref{condition2}). The corresponding profiles of the dielectric components and the photonic band structure are shown in Figs.~\ref{connect}a and b, respectively.

The generalized $1/4$-period screw rotation operator is pseudo-unitary~\cite{mostafazadeh2004Pseudounitary} and obeys $\big(\widetilde{S}_{L/4}\big)^4=\hat{T}_x(L)$, thus its eigenvalue for a Bloch state $\Psi(k_x,0,0)$ on the $k_x$-axis should be a fourth root of $e^{ik_x L}$:
$\widetilde{S}_{L/4}\Psi^{(s)}(k_x,0,0)=s e^{i\frac{k_xL}{4}}\Psi^{(s)}(k_x,0,0)$, where the branch index $s=\pm1,\ \pm i$ classifies the bands in the $k_y=0$ plane into 4 groups. 
Combining time reversal symmetry $\mathcal{T}$ with $\widetilde{S}_{L/4}$, we obtain the pseudo-antiunitary symmetry of the PhC $\hat{\Theta}_{L/4}=\mathcal{T}\,\widetilde{S}_{L/4}$ 
and have (see supplementary information S3)
\begin{gather}
    \hat{\Theta}^2_{L/4}\Psi^{(s)}(0,0,k_z)=s^2\,\Psi^{(s)}(0,0,k_z),\label{Kramers1}
\end{gather}
for the Bloch states on the $k_x=k_y=0$ line ($\Gamma-Z$).
Following a similar derivation as the Kramers theorem, Eq.~\eqref{Kramers1} ensures that, when $s=\pm i$, $\Psi^{(\pm i)}(0,0,k_z)$ and $\widetilde{\Psi}(0,0,k_z)=\hat{\Theta}_{L/4}\Psi^{(\pm i)}(0,0,k_z)$ are two distinct Bloch states degenerate along $\Gamma-Z$, forming Kramers-like nodal lines. Furthermore, since $\widetilde{S}_{L/4}\widetilde{\Psi}(0,0,0)=s^*\,\widetilde{\Psi}(0,0,0)=\mp i\,\widetilde{\Psi}(0,0,0)$,
each band of $s=+i$ along the $k_x$-axis has to intersect with a band of $s=-i$ at $\Gamma$ (\textit{i.e.} the blue dots in Fig.~\ref{connect}c). 

\begin{figure*}[bht]
  \includegraphics[width=0.85\textwidth]{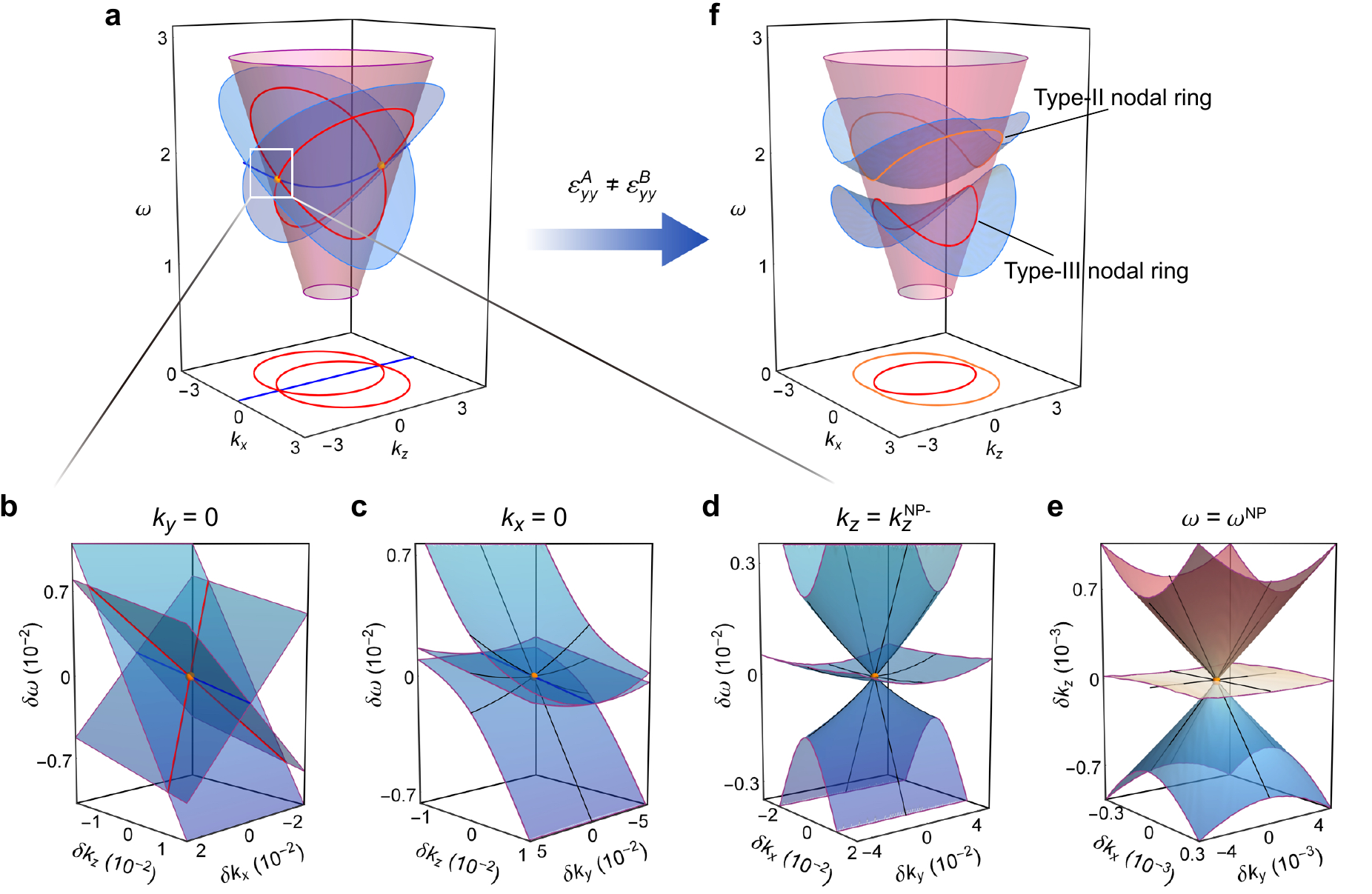}\caption{\label{fig2} Dispersion near the triple nexus points and nodal lines. \textbf{a} 3D band structure of the PhC in Fig.~\ref{fig1} on the high symmetry plane $k_y=0$ illustrating that the \nth{1} $\hat{M}_y$-even band (magenta cone) intersects with a pair of $\hat{M}_y$-odd bands (light blue surfaces) along two red nodal rings, and the two odd bands coincide along the blue Kramers-like NL. The pair of orange dots shows the nexus points of the 3 NLs. \textbf{b}, \textbf{c}, and \textbf{d} Zoomed-in band structures around the NP at $\mathbf{k}^{\mathrm{NP}_-}$ in the $k_y=0$, $k_x=0$, and  $k_z=k_z^{\mathrm{NP}_-}$ sections, respectively, where the vertical coordinate $\delta\omega=\omega-\omega^{\mathrm{NP}}$. \textbf{e} Iso-frequency surfaces around the NP at frequency $\omega^\mathrm{NP}$, where the  black tangential lines are obtained from the \nth{1} order $\mathbf{k}\cdot \mathbf{p}$ Hamiltonian. \textbf{f} Band structure of a PhC with broken hidden symmetry, where $\varepsilon_{yy}$ takes different values in layers A and B of $\varepsilon_{yy}^A=9$ and $\varepsilon_{yy}^B=21$, while all the other components of $\tensor{\varepsilon}_r$ are identical to the case in \textbf{a}.   }
\end{figure*}

At BZ boundaries $\mathbf{k}=(\pm\frac{\pi}{L},0,0)$, the combined symmetry $\hat{\Theta}_{L/2}=\mathcal{T}\hat{S}_{2x}$ guarantees that an arbitrary state $\Psi^{(s)}(\pm\frac{\pi}{L},0,0)$ with branch index $s$ is degenerate with $\widetilde{\Psi}(\pm\frac{\pi}{L},0,0)=\hat{\Theta}_{L/2}\Psi^{(s)}(\pm\frac{\pi}{L},0,0)$. 
Meanwhile, it can be proved that $\widetilde{S}_{L/4}\widetilde{\Psi}(\pm\frac{\pi}{L},0,0)=-i\,s^*\,\widetilde{\Psi}(\pm\frac{\pi}{L},0,0)$, therefore every pair of bands intersecting at the zone boundaries  $k_x=\pm\frac{\pi}{L}$ must have indices of either $s=+1$ and $-i$ or $s=-1$ and $+i$ (see the pairs of bands connecting at the green dots in Fig.~\ref{connect}c).

Fig.~\ref{connect}b exhibits an important difference of the PhC with reduced constraints compared to the PhC consisting of homogeneous layers in Fig.~\ref{fig1}: the \nth{4} and \nth{5} $\hat{M}_y$-odd (even) bands are gapped along $\Gamma-Z$. In fact, since the \nth{4} $\hat{M}_y$-odd (even) band has branch index $s=+1$, Eq.~\eqref{Kramers1} indicates that $\Psi^{(+1)}$ and $\widetilde{\Psi}=\hat{\Theta}_{L/4}\Psi^{(+1)}$ can be the same state. To achieve higher band connectivity along $\Gamma-X$, we need the components of the permittivity to have a smaller fractional period $L/n$ ($n>4$). 
The layered PhC in Fig.~\ref{fig1} can be viewed as the limiting case of infinitesimal fractional periodicity ($n\to\infty$).

\textbf{\textit{Photonic band connectivity}}\\
Dielectric PhCs have a universal feature that there are always two gapless photonic bands emerging from the singular point $\omega=|\mathbf{k}|=0$, around which the Bloch modes on the two gapless bands are transverse plane waves in the long-wavelength limit\cite{watanabe2018Space}. 
If the PhCs further meet the conditions of Eqs.~\eqref{condition1} and \eqref{condition2}, then the eigenvalues of $\widetilde{S}_{L/4}$ for the two lowest bands connected to zero frequency are both equal to $-e^{i\frac{k_xL}{4}}$, namely the first $\hat{M}_y$ even and odd bands have the same branch index $s=-1$ (see supplementary information S4). Starting from the first $\hat{M}_y$-even (odd) band along $\Gamma-X$, $\widetilde{S}_{L/4}$ symmetry ensures that at least 4 bands with branch indices $-1\rightarrow +i\rightarrow -i \rightarrow +1$ (counting from the bottom) concatenate successively at the Kramers-like degeneracies at $k_x=\frac{\pi}{L}$ and $k_x=0$, as shown in Fig.~\ref{connect}c. Consequently, 
the minimal band connectivity (MBC) along $\Gamma-X$ is 8 for bands connected to zero frequency, which is beyond the prediction ($\mathrm{MBC}=4$) made by only considering the twofold screw symmetry $\hat{S}_{2x}$~\cite{watanabe2018Space}. $\mathrm{MBC}=8$ implies that the lowest 4 even and lowest 4 odd bands inevitably intersect at least 3 times (the red dots in Figs.~\ref{connect}b,c and in Fig.~\ref{fig1}b) along the line segment $\Gamma-X$, and therefore the unique photonic band connectivity enforces the emergence of the two red nodal rings shown in Fig.~\ref{fig1}c. For bands not connected to zero frequency, $\widetilde{S}_{L/4}$ symmetry leads to $MBC=4$ along $\Gamma-X$, which is also twice the result determined solely by the space group.

In addition, Fig.~\ref{connect}b shows that all bands along $\Gamma-Z$ tend towards linear dispersion as $k_z\to \infty$. Indeed, it can be rigorously proved  that, in an $\hat{M}_y$-symmetric dielectric PhC with continuous translational symmetry along the $z$-axis, all $\hat{M}_y$-even and all $\hat{M}_y$-odd bands have identical asymptotic group velocities in the $z$ direction, respectively  (see details in supplementary information S5),
\begin{equation}
    \hspace{0pt}\lim_{k_z\to\infty}\frac{d\omega^\mathrm{even}}{dk_z}=\frac{c}{\sqrt{\varepsilon^\mathrm{max}_{xx}}},\quad\ 
    \lim_{k_z\to\infty}\frac{d\omega^\mathrm{odd}}{dk_z}=\frac{c}{\sqrt{\varepsilon^\mathrm{max}_{yy}}},
\end{equation}
where $\varepsilon^\mathrm{max}_{ii}$ denotes the maximum value of $\varepsilon_{ii}$ $(i=x,y)$ in the PhC, and $c$ is the speed of light in vacuum, as demonstrated by the numerical results in Fig.~\ref{connect}d. This result can be understood from the physical picture that the EM fields tend to concentrate in the regions of high refractive index. In the short wavelength limit ($k_z\to\infty$), all the fields will be localized at the maximal permittivity positions, and hence, only $\varepsilon^\mathrm{max}_{xx}$ and $\varepsilon^\mathrm{max}_{yy}$ determine the asymptotic dispersion. This special property indicates that, as long as $\varepsilon^\mathrm{max}_{xx}\neq\varepsilon^\mathrm{max}_{yy}$, \textit{e.g.}, $\varepsilon^\mathrm{max}_{xx}<\varepsilon^\mathrm{max}_{yy}\  (\varepsilon^\mathrm{max}_{xx}>\varepsilon^\mathrm{max}_{yy})$, the \nth{1} $\hat{M}_y$-even (odd) band has to intersect with the \nth{2} and \nth{3} $\hat{M}_y$-odd (even) bands, \textit{i.e.}, the \nth{1} $\hat{M}_y$-odd (even) Kramers NL, on $\Gamma-Z$ (the orange dot in Fig.~\ref{connect}b) owing to their different asymptotic group velocities. 

As a notable consequence, almost any layer-stacked dielectric PhC respecting $\hat{M}_y$ and the hidden symmetry $\widetilde{S}_{L/4}$ (\textit{e.g.} the AB-layer-stacked PhC in Fig.~\ref{fig1}) 
must carry a pair of triply degenerate nodes as the nexuses of the two $\hat{M}_y$-symmetry-protected nodal rings (red rings in Fig.~\ref{fig1}c) and \nth{1} Kramers-like NL (blue line  in Fig.~\ref{fig1}c) for the bands connecting to $\omega=|\mathbf{k}|=0$,  as displayed in Fig.~\ref{fig1}c. More examples with different parameters of the PhC are given in supplementary information S6, where we can see that the NPs always appear unless the asymptotic group velocities of even and odd bands are accidentally identical, namely $\varepsilon^\mathrm{max}_{xx}=\varepsilon^\mathrm{max}_{yy}$.
Nevertheless, these exceptions only form a subset of measure zero for all possible parameters of the PhCs.

We are also aware that the linear asymptotic dispersion along the $z$ direction causes infinitely many bands with opposite mirror parities to intersect as long as $\varepsilon^\mathrm{max}_{xx}\neq\varepsilon^\mathrm{max}_{yy}$, hence forming not only infinitely many threefold NPs but also infinitely many fourfold NPs (see the purple dots in Fig.~\ref{connect}c and supplementary information S6 for the typical dispersion around a fourfold NP). Hereinafter, we will focus on the lowest pair of  triply degenerate NPs in the AB-layer-stacked PhC shown in Fig.~\ref{fig1} and investigate the band dispersion near the NPs.

\textbf{\textit{Triply degenerate nexus points}}\\
Fig.~\ref{fig2}a displays the 3D band structure near the 3 nodal lines in the $k_y=0$ plane corresponding to the PhC in Fig.~\ref{fig1}, where the magenta cone, denoting the \nth{1} $\hat{M}_y$-even band,  cuts across the \nth{2} and \nth{3} $\hat{M}_y$-odd bands (two light blue surfaces) along the two nodal rings. Meanwhile, the two odd bands connect at the Kramers-like NL along $k_x=0$. As a result, the three NLs intersect at a pair of triple NPs (orange dots) with Bloch wave vectors $\mathbf{k}^\mathrm{NP_\pm}=(0,0,\pm\frac{2\pi}{L}\sqrt{\frac{\varepsilon_{xx}}{\varepsilon_{yy}-\varepsilon_{xx}}})$ and frequency $\omega^\mathrm{NP}=\frac{2\pi c}{L\sqrt{\varepsilon_{yy}-\varepsilon_{xx}}}$. 
In terms of the $\mathbf{k}\cdot\mathbf{p}$ perturbation approach, the effective Hamiltonian around the NPs up to the linear order of $\delta\mathbf{k}=(\delta k_x,\delta k_y,\delta k_z)=\mathbf{k}-\mathbf{k}^\mathrm{NP_\pm}$ is given by
\begin{equation}\label{kdotp hamiltonian}
\begin{split}
   \hspace{-5pt}\hat{H}^\pm_\mathrm{NP}=&\begin{pmatrix}
    v_x\delta k_x\pm v^\mathrm{odd}_z\delta k_z & \displaystyle\frac{-i\,v_y}{\sqrt{2}}\delta k_y & 0 \\
    \displaystyle\frac{i\,v_y}{\sqrt{2}}\delta k_y & \pm v^\mathrm{even}_z\delta k_z & \displaystyle\frac{-i\,v_y}{\sqrt{2}}\delta k_y\\
    0 & \displaystyle\frac{i\,v_y}{\sqrt{2}}\delta k_y & -v_x\delta k_x\pm v^\mathrm{odd}_z\delta k_z
    \end{pmatrix}\\
    =&v_x\hat{S}_z\delta k_x+v_y\hat{S}_y\delta k_y\pm\left[q_z\hat{Q}_{zz}+v_{z0}\hat{I}\right]\delta k_z,
\end{split}
\end{equation}
where 
$v_x$, $v_y$, $v^\mathrm{odd}_z$, and $v^\mathrm{even}_z$ are the group velocities along the corresponding directions, $\hat{S}_i$ ($i=x,y,z$) denote the spin-1 operators, $\hat{Q}_{zz}=(\hat{S}_z)^2-\sum_{i}(\hat{S}_i)^2/3$ is one of the spin-1 quadrupolar operators~\cite{hu2018Topological}, $q_z=v^\mathrm{odd}_z-v^\mathrm{even}_z$, and $v_{z0}=\frac{1}{3}(2v^\mathrm{odd}_z+v^\mathrm{even}_z)$  (see supplementary information S7 for details).

The bands around NPs present unusual anisotropic dispersion, as depicted by the band structures on different sections passing through an NP ($\mathbf{k}^{\mathrm{NP}_-}$) in Figs.~\ref{fig2}b-d. 
In the $k_y=0$ section, Fig.~\ref{fig2}b  shows that all three bands cross linearly along the NLs. 
In Fig.~\ref{fig2}c, the band structure in the $k_x=0$ section resembles the dispersion of the so-called type-II triply degenerate point~\cite{chang2017Nexus}, where two bands sandwich and contact the third band along the Kramers-like NL, and all their group velocities ($-v^\mathrm{odd}_z$ for the NL and $-v^\mathrm{even}_z$ for the singlet) are of the same sign in the $z$ direction.
Remarkably, in the $k_{z}={k}_z^\mathrm{NP_-}=-\frac{2\pi}{L}\sqrt{\frac{\varepsilon_{xx}}{\varepsilon_{yy}-\varepsilon_{xx}}}$ section, Fig.~\ref{fig2}d shows that two conical bands intersect with an almost  flat band at an NP,  manifesting as a  2D anisotropic Dirac-like cone described by the 2D spin-1 Hamiltonian $\hat{H}_\mathrm{NP}(\delta k_z=0)=v_x\hat{S}_z\delta k_x+v_y\hat{S}_y\delta k_y$ \cite{green2010Isolated,huang2011Dirac}. 
Furthermore, as shown in Fig.~\ref{fig2}e, the iso-frequency surfaces around the NP at $\omega_\mathrm{NP}$ also disperse exactly as a 2D spin-1 Dirac-like cone in the $xy$ plane, if $\delta k_z$ is regarded as a pseudo-frequency. 
This property implies that the NPs in layer-stacked PhCs can be used to realize the novel physical effects associated with 2D spin-1 dispersion.

The unique band topology of the NPs in our system demonstrates that they belong to a new kind of threefold nodal point, different from all isolated triple  points carrying topological charge~\cite{bradlyn2016Dirac,saba2017Group,hu2018Topological,yang2019Topological,zhang2018DoubleWeyl}. In fact, the charge of an NP cannot be defined, because, for an arbitrary closed surface enclosing the NP, the gap between any 2 of the 3 bands on the surface must shut at the points where the NLs between the 2 bands pierce the surface~\cite{zhu2016Triple,chang2017Nexus}. Nonetheless, since the Berry fluxes are confined on the NLs in a $\mathcal{PT}$ symmetric system, the NPs are the terminations of Berry flux strings, and can thus be regarded as a different kind of magnetic monopole, other than Weyl points, in momentum space~\cite{cornwall1999Center,volovik2000Monopoles,heikkila2015Nexus}.

\textbf{\textit{Type-II and type-III nodal rings}}\\
The simple AB-layer PhC can be the parent structures of other fancy topological features.
In Fig.~\ref{fig2}f, we let $\varepsilon_{yy}^A\neq\varepsilon_{yy}^B$ to observe the process of hidden symmetry $\widetilde{S}_{L/4}$ breaking. Without protection by $\hat{\Theta}_{L/4}=\mathcal{T}\,\widetilde{S}_{L/4}$, the Kramers-like degeneracy along $k_x=0$ between the \nth{2} and \nth{3} $\hat{M}_y$-odd bands is lifted, and the pair of NPs disappears. 
Consequently, the original crossed nodal rings split into two new isolated rings. For the upper nodal ring (orange), the two degenerate bands are significantly tilted in the  direction $\hat{\mathbf{k}}_\perp$ perpendicular to the ring such that both their perpendicular group velocities $v^\mathrm{even/odd}_\perp=\nabla_\mathbf{k}\omega^\mathrm{even/odd}\cdot\hat{\mathbf{k}}_\perp$ always have the same sign at any point on the ring, forming a  type-\RNum{2} nodal ring~\cite{gao2018Class,zhang2018Hybrid,he2018TypeII}. 
As the lower $\hat{M}_y$-odd band has a saddle-shaped dispersion, both type-\RNum{1} points (\textit{i.e.}, $v^\mathrm{odd}_\perp$ and $v^\mathrm{even}_\perp$ of opposite signs) and type-\RNum{2} points coexist on the lower nodal ring (red), and such band crossings are referred to as a type-\RNum{3} nodal ring~\cite{gao2018Class} or a hybrid nodal ring~\cite{zhang2018Hybrid}.

\begin{figure*}[ht!]
\includegraphics[width=0.8\textwidth]{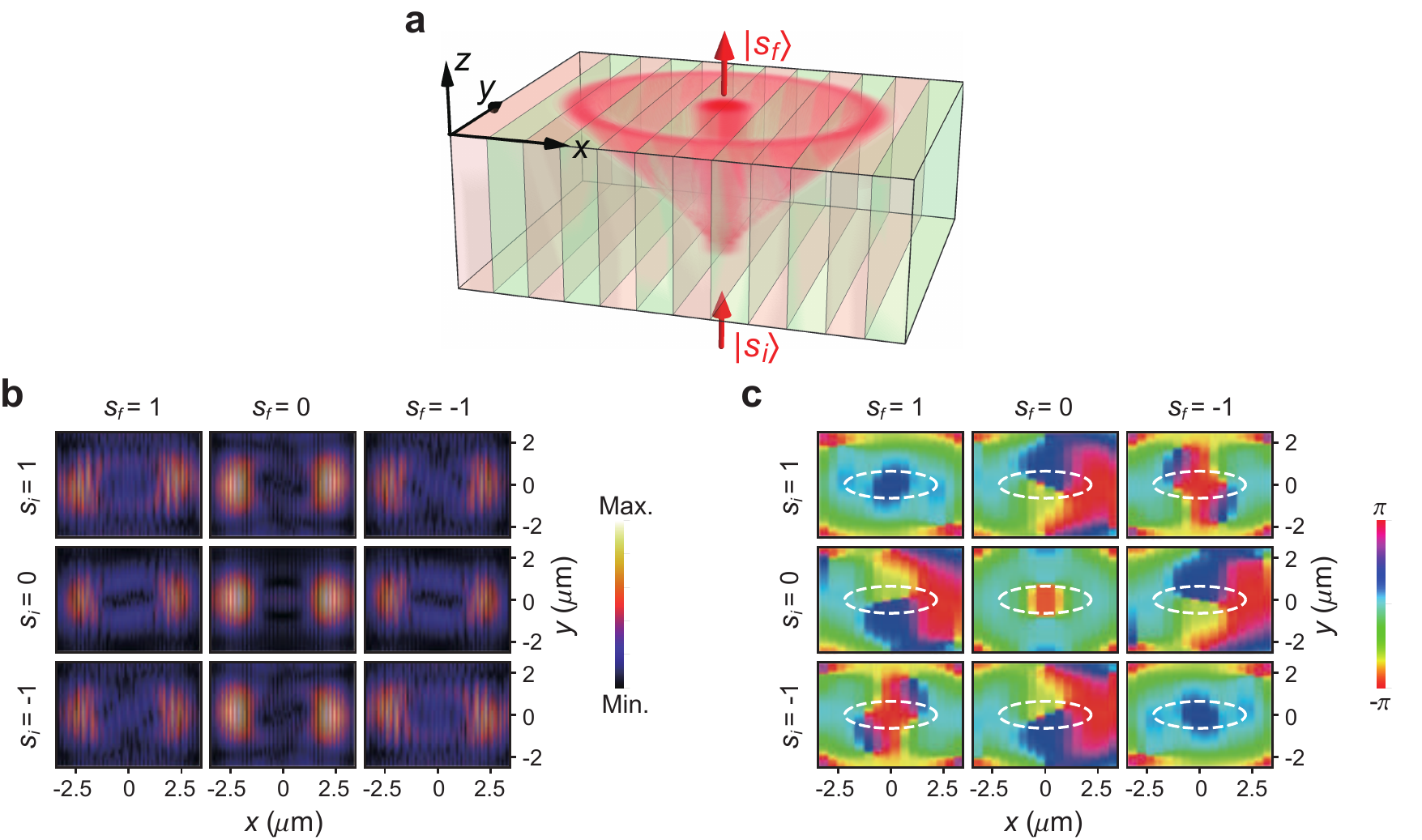}\caption{\label{fig4} Spin-1 conical diffraction at nexus point. \textbf{a} Schematic of spin-1 conical diffraction for a monochromatic beam with frequency $\omega=\omega^\mathrm{NP}$ incident on the PhC along the $z$-axis. The hollow cone illustrates the envelope of the trajectories corresponding to the wave components on the two conical iso-frequency surfaces around the NP. The straight light beam in the PhC corresponds to the wave components on the flat iso-frequency surface. 
Simulated \textbf{b} intensity and \textbf{c} phase (average of each period in the $x$ direction) distribution of different output spin components (column indices) on the horizontal $xy$ plane as the incident spin eigenstate of $\hat{S}_x$  (row indices) varies, where the parameters of the PhC are $\varepsilon_1=\varepsilon_2=12$, $\varepsilon_3=1$, $\theta=\pi/6$, and $L=0.43\,\mathrm{\mu m}$. The vacuum wavelength of the incident beam is $0.713\,\mathrm{\mu m}$. The distance between the incident and output planes is $4\,\mathrm{\mu m}$.}
\end{figure*}

\textbf{\textit{Spin-1 conical diffraction}}\\ 
It is well known that light beams travelling along the optical axes in  biaxial crystals will undergo conical diffraction~\cite{berry2007Chapter}. The conical diffraction phenomenon is actually a generic scattering effect for twofold degenerate Dirac points~\cite{peleg2007Conical}, and should also occur for light scattered by almost any linearly crossing point on the nodal rings in our system. Specifically, when an incident light beam has frequency and momentum that match a certain point on the  nodal rings, 
the refractive waves spread into a hollow cone in the PhC, and at the same time,  the polarizations circling  the cone trace out a great circle on the Poincar\'e sphere, manifesting the quantized $\pi$ Berry phase encircling the nodal ring. 

In contrast, the diffraction at the triple NPs appears strikingly different from that at other points on the NLs. Since the iso-frequency surfaces around each NP form a spin-1 Dirac-like cone (Fig.~\ref{fig2}e), the monochromatic dynamics at an NP, \textit{e.g.} $\mathbf{k}^{\mathrm{NP}_+}$, is effectively described by a Schr\"odinger equation with the 2D spin-1 Hamiltonian $\hat{H}(\delta\mathbf{k}_{xy})=v_x\hat{S}_z\delta k_x+\tilde{v}_y\hat{S}_y\delta k_y$ (here $\tilde{v}_y=\sqrt{v^\mathrm{even}_z/v^\mathrm{odd}_z}\,v_y$):
\begin{equation}\label{spin1H}
i\,v^\mathrm{odd}_z \frac{\partial}{\partial z} \left|\psi\right>=\hat{H}(\delta\mathbf{k}_{xy})\left|\psi\right>,
\end{equation} 
where the $z$-coordinate serves as pseudo-time.
Therefore, waves incident at the NPs should experience unconventional spin-1 conical diffraction~\cite{diebel2016Conical} rather than the spin-1/2-type diffraction at ordinary diabolic points.  
A schematic of spin-1 conical diffraction is shown in Fig.~\ref{fig4}a, where a light beam with frequency $\omega^\mathrm{NP}$ is incident along the $z$-axis in the PhC, and its wave vector spectrum concentrates near an NP. Since the NP is a singularity of group velocity for wave components on the two conical bands, these components will spread over a conical surface, whereas the components on the flat band will propagate straight along the $z$-axis. More interestingly, if the initial state of the beam is  an eigenstate $|s_i\rangle$ of $\hat{S}_x$ with spin quantum number $s_i\in\{-1,0,1\}$, then the spin-1 character is inherent in the transition amplitude from $|s_i\rangle$ to another eigenstate $|s_f\rangle$ of $\hat{S}_x$ as the output of the diffraction process (see supplementary information S8):
\begin{equation}\label{transition}
    \big\langle s_f \big|
    e^{-i\hat{H}(\delta\mathbf{k}_{xy})z/{v^\mathrm{odd}_z}}\big|s_i\big\rangle\propto \exp\left[i(s_f-s_i)\phi(\delta\mathbf{k}_{xy})\right],
\end{equation}
where $\phi(\delta\mathbf{k}_{xy})$ denotes the polar angle of $(v_x\delta k_x+i\,\tilde{v}_y\delta k_y)$. Eq.~\eqref{transition} shows that the phase of the output field winds $l=s_f-s_i$ times around $\delta\mathbf{k}_{xy}=0$ in momentum space. Because the trajectories of the wave components encircling $\delta\mathbf{k}_{xy}=0$ also wrap around the $z$-axis in real space, the output field projected onto $|s_f\rangle$ will generate an optical vortex on the ring-shaped section of the diffractive cone, and the charge of the vortex, $l=s_f-s_i\in\{0,\pm1,\pm2\}$, is determined by the difference in the spin quantum number between the final and initial spin states~\cite{diebel2016Conical}, which essentially reflects the conservation of the total generalized angular momentum for the spin-1 Hamiltonian (see supplementary information S8).

We have performed full-wave simulations for the 9 possible combinations of input and output spin eigenstates. The intensities and phases of the output fields on a horizontal plane are shown in Figs.~\ref{fig4}b and c, respectively, where the asymmetric intensity distributions originate from the anisotropy of the Dirac-like dispersion. In each panel of Fig.~\ref{fig4}c , the white dashed ring corresponds to the section of the diffractive cone in the geometric optics approximation, and the winding number of the phase along the ring defines the charge of the optical vortex. We can see that the winding numbers agree with the theoretical predictions $l=s_f-s_i$ in all panels, so we have demonstrated that spin-1 conical diffraction occurs in the PhC, which also indicates that the NPs can be used as a new route for studying 2D spin-1 dynamics.

\textbf{Discussion}\\[3pt]
We have discovered that a class of simple layer-stacked PhCs manifests a hidden symmetry of Maxwell's equations which directly influences the connectivity of photonic bands and engenders a pair of triply degenerate NPs where three symmetry-enforced NLs intersect. 
These photonic NPs not only are worthy of theoretical investigation as a novel kind of magnetic monopole that terminates Berry flux strings in momentum space, but also induce exotic bulk transport effects in the PhC, which may lead to prospective applications. In particular, we  found that the unusual spin-1 Dirac-like dispersion of the iso-frequency surfaces near an NP can induce spin-1 conical diffraction of optical beams, which 
can be used to generate optical vortices with a maximum topological charge of $2$. 

Our work takes the first step towards new research directions. On the one hand, the hidden symmetry of Maxwell's equations reveals a novel mechanism for realizing protected degeneracies unique to photonic bands. 
The hidden symmetry here stems from the fractional periods of different components of $\tensor{\varepsilon}_r$, which reflects a geometric property of the PhC but cannot be described by the conventional space groups. It would be fundamentally significant to develop a generalized space group theory including such symmetries for photonic systems. 
On the other hand, our proposed PhC consists of a single anisotropic dielectric material, but the band peculiarity comes entirely from the nontrivial periodic rotation of the optical axes inside the material. Although recent studies have shown that artificial gauge fields~\cite{liu2015Gauge,chen2019NonAbelian},  Pancharatnam-Berry phases~\cite{jisha2019SelfTrapping}, and synthetic spin-orbit interactions~\cite{rechcinska2019Engineering} for light can be achieved by arranging the dielectric polarization, there are still little works that study the photonic band topology of the PhCs made of anisotropic dielectrics. Our results suggest that the intrinsic material anisotropy has irreplaceable features compared with to structural anisotropy, and anisotropic dielectrics, such as liquid crystals~\cite{rechcinska2019Engineering,jisha2019SelfTrapping}, could become a platform for probing the unique topological effects of EM fields.

\textbf{Materials and Methods}\\[3pt]
The methods for calculating the band structures include the transfer matrix approach for the PhC made of homogeneous anisotropic dielectric layers, with the results shown in Figs.~\ref{fig1} and \ref{fig2} (see supplementary information S1 for details), and the plane wave expansion method for the PhC made of generic inhomogeneous dielectrics, with the results shown in Fig.~\ref{connect}. The numerical results of the spin-1 conical diffraction in Fig.~\ref{fig4} are simulated using the commercial software COMSOL Multiphysics. The input state with a certain spin quantum number is achieved by setting the field distributions on the lower boundary of the PhC as the superposition of the exact Bloch eigenstates according to the  Hamiltonian in Eq.~\eqref{kdotp hamiltonian}. The output fields projected onto the spin eigenstates are averaged along the  direction in each period (a pair of AB layers), and the phase distributions shown in Fig.~\ref{fig4}c correspond to these averaged fields.\\

\textbf{Acknowledgements}\\
We thank Profs. Yun Lai, Jie Luo, Neng Wang, and Zhao-Qing Zhang for the helpful discussions. 
This work is supported by the Natural National Science Foundation (NSFC) (Grant No. 11874026), and the Research Grants Council of Hong Kong, China (Grant Nos. AoE/P-02/12 and 16304717). 

\textbf{Data availability}\\
 The authors declare that all data supporting the findings of this study are available from the corresponding authors upon reasonable request. 
 
\textbf{Conflict of interest}\\
The authors declare that they have no conflicts of interest.

\textbf{Contributions}\\
Y.C., and R.-Y. Z. conceived the original idea. Z.X. discovered the triple NPs. R.-Y.Z., Z.X., Y.C., and R.Y. developed the theory. Z.X. R.-Y.Z., and Y.C. carried out the numerical simulations. R.-Y.Z., Z.X., Y.C., and C.T.C wrote the manuscript. C.T.C. and Y.C. supervised the project. All authors were involved in the analysis and discussion of the results.\\


\def\bibsection{\section*{}}
\textbf{References}\\
\vspace{-50pt}
\bibliographystyle{naturemag}
\bibliography{references}

\begin{thebibliography}{10}
\expandafter\ifx\csname url\endcsname\relax
  \def\url#1{\texttt{#1}}\fi
\expandafter\ifx\csname urlprefix\endcsname\relax\def\urlprefix{URL }\fi
\providecommand{\bibinfo}[2]{#2}
\providecommand{\eprint}[2][]{\url{#2}}

\bibitem{wan2011Topological}
\bibinfo{author}{Wan, X.}, \bibinfo{author}{Turner, A.~M.},
  \bibinfo{author}{Vishwanath, A.} \& \bibinfo{author}{Savrasov, S.~Y.}
\newblock \bibinfo{title}{Topological semimetal and {{Fermi}}-arc surface
  states in the electronic structure of pyrochlore iridates}.
\newblock \emph{\bibinfo{journal}{Phys. Rev. B}} \textbf{\bibinfo{volume}{83}},
  \bibinfo{pages}{205101} (\bibinfo{year}{2011}).

\bibitem{liu2014Discoverya}
\bibinfo{author}{Liu, Z.~K.} \emph{et~al.}
\newblock \bibinfo{title}{Discovery of a {{Three}}-{{Dimensional Topological
  Dirac Semimetal}}, {{Na3Bi}}}.
\newblock \emph{\bibinfo{journal}{Science}} \textbf{\bibinfo{volume}{343}},
  \bibinfo{pages}{864--867} (\bibinfo{year}{2014}).

\bibitem{xu2015Discovery}
\bibinfo{author}{Xu, S.-Y.} \emph{et~al.}
\newblock \bibinfo{title}{Discovery of a {{Weyl}} fermion semimetal and
  topological {{Fermi}} arcs}.
\newblock \emph{\bibinfo{journal}{Science}} \textbf{\bibinfo{volume}{349}},
  \bibinfo{pages}{613--617} (\bibinfo{year}{2015}).

\bibitem{armitage2018Weyl}
\bibinfo{author}{Armitage, N.~P.}, \bibinfo{author}{Mele, E.~J.} \&
  \bibinfo{author}{Vishwanath, A.}
\newblock \bibinfo{title}{Weyl and {{Dirac}} semimetals in three-dimensional
  solids}.
\newblock \emph{\bibinfo{journal}{Rev. Mod. Phys.}}
  \textbf{\bibinfo{volume}{90}}, \bibinfo{pages}{015001}
  (\bibinfo{year}{2018}).

\bibitem{lu2013Weyl}
\bibinfo{author}{Lu, L.}, \bibinfo{author}{Fu, L.},
  \bibinfo{author}{Joannopoulos, J.~D.} \& \bibinfo{author}{Solja{\v c}i{\'c},
  M.}
\newblock \bibinfo{title}{Weyl points and line nodes in gyroid photonic
  crystals}.
\newblock \emph{\bibinfo{journal}{Nat. Photonics}}
  \textbf{\bibinfo{volume}{7}}, \bibinfo{pages}{294--299}
  (\bibinfo{year}{2013}).

\bibitem{lu2015Experimental}
\bibinfo{author}{Lu, L.} \emph{et~al.}
\newblock \bibinfo{title}{Experimental observation of {{Weyl}} points}.
\newblock \emph{\bibinfo{journal}{Science}} \textbf{\bibinfo{volume}{349}},
  \bibinfo{pages}{622--624} (\bibinfo{year}{2015}).

\bibitem{chen2016Photonic}
\bibinfo{author}{Chen, W.-J.}, \bibinfo{author}{Xiao, M.} \&
  \bibinfo{author}{Chan, C.~T.}
\newblock \bibinfo{title}{Photonic crystals possessing multiple {{Weyl}} points
  and the experimental observation of robust surface states}.
\newblock \emph{\bibinfo{journal}{Nat. Commun.}} \textbf{\bibinfo{volume}{7}},
  \bibinfo{pages}{13038} (\bibinfo{year}{2016}).

\bibitem{noh2017Experimental}
\bibinfo{author}{Noh, J.} \emph{et~al.}
\newblock \bibinfo{title}{Experimental observation of optical {{Weyl}} points
  and {{Fermi}} arc-like surface states}.
\newblock \emph{\bibinfo{journal}{Nat. Phys.}} \textbf{\bibinfo{volume}{13}},
  \bibinfo{pages}{611--617} (\bibinfo{year}{2017}).

\bibitem{chang2017Multiplea}
\bibinfo{author}{Chang, M.-L.}, \bibinfo{author}{Xiao, M.},
  \bibinfo{author}{Chen, W.-J.} \& \bibinfo{author}{Chan, C.~T.}
\newblock \bibinfo{title}{Multiple {{Weyl}} points and the sign change of their
  topological charges in woodpile photonic crystals}.
\newblock \emph{\bibinfo{journal}{Phys. Rev. B}} \textbf{\bibinfo{volume}{95}},
  \bibinfo{pages}{125136} (\bibinfo{year}{2017}).

\bibitem{yang2018Ideal}
\bibinfo{author}{Yang, B.} \emph{et~al.}
\newblock \bibinfo{title}{Ideal {{Weyl}} points and helicoid surface states in
  artificial photonic crystal structures}.
\newblock \emph{\bibinfo{journal}{Science}} \textbf{\bibinfo{volume}{359}},
  \bibinfo{pages}{1013--1016} (\bibinfo{year}{2018}).

\bibitem{wang2016Threedimensional}
\bibinfo{author}{Wang, H.}, \bibinfo{author}{Xu, L.}, \bibinfo{author}{Chen,
  H.} \& \bibinfo{author}{Jiang, J.-H.}
\newblock \bibinfo{title}{Three-dimensional photonic {{Dirac}} points
  stabilized by point group symmetry}.
\newblock \emph{\bibinfo{journal}{Phys. Rev. B}} \textbf{\bibinfo{volume}{93}},
  \bibinfo{pages}{235155} (\bibinfo{year}{2016}).

\bibitem{wang2017TypeII}
\bibinfo{author}{Wang, H.-X.}, \bibinfo{author}{Chen, Y.},
  \bibinfo{author}{Hang, Z.~H.}, \bibinfo{author}{Kee, H.-Y.} \&
  \bibinfo{author}{Jiang, J.-H.}
\newblock \bibinfo{title}{Type-{{II Dirac}} photons}.
\newblock \emph{\bibinfo{journal}{npj Quantum Mater.}}
  \textbf{\bibinfo{volume}{2}}, \bibinfo{pages}{54} (\bibinfo{year}{2017}).

\bibitem{guo2019Observation}
\bibinfo{author}{Guo, Q.} \emph{et~al.}
\newblock \bibinfo{title}{Observation of {{Three}}-{{Dimensional Photonic Dirac
  Points}} and {{Spin}}-{{Polarized Surface Arcs}}}.
\newblock \emph{\bibinfo{journal}{Phys. Rev. Lett.}}
  \textbf{\bibinfo{volume}{122}}, \bibinfo{pages}{203903}
  (\bibinfo{year}{2019}).

\bibitem{weng2015Topological}
\bibinfo{author}{Weng, H.} \emph{et~al.}
\newblock \bibinfo{title}{Topological node-line semimetal in three-dimensional
  graphene networks}.
\newblock \emph{\bibinfo{journal}{Phys. Rev. B}} \textbf{\bibinfo{volume}{92}},
  \bibinfo{pages}{045108} (\bibinfo{year}{2015}).

\bibitem{chan2016MathrmCa}
\bibinfo{author}{Chan, Y.-H.}, \bibinfo{author}{Chiu, C.-K.},
  \bibinfo{author}{Chou, M.~Y.} \& \bibinfo{author}{Schnyder, A.~P.}
\newblock \bibinfo{title}{{$\mathrm{Ca}_3\mathrm{P}_2$} and other topological
  semimetals with line nodes and drumhead surface states}.
\newblock \emph{\bibinfo{journal}{Phys. Rev. B}} \textbf{\bibinfo{volume}{93}},
  \bibinfo{pages}{205132} (\bibinfo{year}{2016}).

\bibitem{fang2016Topological}
\bibinfo{author}{Fang, C.}, \bibinfo{author}{Weng, H.}, \bibinfo{author}{Dai,
  X.} \& \bibinfo{author}{Fang, Z.}
\newblock \bibinfo{title}{Topological nodal line semimetals}.
\newblock \emph{\bibinfo{journal}{Chin. Phys. B}}
  \textbf{\bibinfo{volume}{25}}, \bibinfo{pages}{117106}
  (\bibinfo{year}{2016}).

\bibitem{kim2015Dirac}
\bibinfo{author}{Kim, Y.}, \bibinfo{author}{Wieder, B.~J.},
  \bibinfo{author}{Kane, C.~L.} \& \bibinfo{author}{Rappe, A.~M.}
\newblock \bibinfo{title}{Dirac {{Line Nodes}} in {{Inversion}}-{{Symmetric
  Crystals}}}.
\newblock \emph{\bibinfo{journal}{Phys. Rev. Lett.}}
  \textbf{\bibinfo{volume}{115}}, \bibinfo{pages}{036806}
  (\bibinfo{year}{2015}).

\bibitem{yu2015Topological}
\bibinfo{author}{Yu, R.}, \bibinfo{author}{Weng, H.}, \bibinfo{author}{Fang,
  Z.}, \bibinfo{author}{Dai, X.} \& \bibinfo{author}{Hu, X.}
\newblock \bibinfo{title}{Topological {{Node}}-{{Line Semimetal}} and {{Dirac
  Semimetal State}} in {{Antiperovskite}} {$\mathrm{Cu}_3\mathrm{PdN}$}}.
\newblock \emph{\bibinfo{journal}{Phys. Rev. Lett.}}
  \textbf{\bibinfo{volume}{115}}, \bibinfo{pages}{036807}
  (\bibinfo{year}{2015}).

\bibitem{gao2018Class}
\bibinfo{author}{Gao, Y.} \emph{et~al.}
\newblock \bibinfo{title}{A class of topological nodal rings and its
  realization in carbon networks}.
\newblock \emph{\bibinfo{journal}{Phys. Rev. B}} \textbf{\bibinfo{volume}{97}},
  \bibinfo{pages}{121108} (\bibinfo{year}{2018}).

\bibitem{zhang2018Hybrid}
\bibinfo{author}{Zhang, X.} \emph{et~al.}
\newblock \bibinfo{title}{Hybrid nodal loop metal: {{Unconventional}}
  magnetoresponse and material realization}.
\newblock \emph{\bibinfo{journal}{Phys. Rev. B}} \textbf{\bibinfo{volume}{97}},
  \bibinfo{pages}{125143} (\bibinfo{year}{2018}).

\bibitem{he2018TypeII}
\bibinfo{author}{He, J.}, \bibinfo{author}{Kong, X.}, \bibinfo{author}{Wang,
  W.} \& \bibinfo{author}{Kou, S.-P.}
\newblock \bibinfo{title}{Type-{{II}} nodal line semimetal}.
\newblock \emph{\bibinfo{journal}{New J. Phys.}} \textbf{\bibinfo{volume}{20}},
  \bibinfo{pages}{053019} (\bibinfo{year}{2018}).

\bibitem{kawakami2017Symmetryguaranteed}
\bibinfo{author}{Kawakami, T.} \& \bibinfo{author}{Hu, X.}
\newblock \bibinfo{title}{Symmetry-guaranteed nodal-line semimetals in an fcc
  lattice}.
\newblock \emph{\bibinfo{journal}{Phys. Rev. B}} \textbf{\bibinfo{volume}{96}},
  \bibinfo{pages}{235307} (\bibinfo{year}{2017}).

\bibitem{yan2018Experimental}
\bibinfo{author}{Yan, Q.} \emph{et~al.}
\newblock \bibinfo{title}{Experimental discovery of nodal chains}.
\newblock \emph{\bibinfo{journal}{Nat. Phys.}} \textbf{\bibinfo{volume}{14}},
  \bibinfo{pages}{461--464} (\bibinfo{year}{2018}).

\bibitem{gao2018Experimental}
\bibinfo{author}{Gao, W.} \emph{et~al.}
\newblock \bibinfo{title}{Experimental observation of photonic nodal line
  degeneracies in metacrystals}.
\newblock \emph{\bibinfo{journal}{Nat. Commun.}} \textbf{\bibinfo{volume}{9}},
  \bibinfo{pages}{950} (\bibinfo{year}{2018}).

\bibitem{xia2019Observation}
\bibinfo{author}{Xia, L.} \emph{et~al.}
\newblock \bibinfo{title}{Observation of {{Hourglass Nodal Lines}} in
  {{Photonics}}}.
\newblock \emph{\bibinfo{journal}{Phys. Rev. Lett.}}
  \textbf{\bibinfo{volume}{122}}, \bibinfo{pages}{103903}
  (\bibinfo{year}{2019}).

\bibitem{bradlyn2016Dirac}
\bibinfo{author}{Bradlyn, B.} \emph{et~al.}
\newblock \bibinfo{title}{Beyond {{Dirac}} and {{Weyl}} fermions:
  {{Unconventional}} quasiparticles in conventional crystals}.
\newblock \emph{\bibinfo{journal}{Science}} \textbf{\bibinfo{volume}{353}},
  \bibinfo{pages}{aaf5037} (\bibinfo{year}{2016}).

\bibitem{saba2017Group}
\bibinfo{author}{Saba, M.}, \bibinfo{author}{Hamm, J.~M.},
  \bibinfo{author}{Baumberg, J.~J.} \& \bibinfo{author}{Hess, O.}
\newblock \bibinfo{title}{Group {{Theoretical Route}} to {{Deterministic Weyl
  Points}} in {{Chiral Photonic Lattices}}}.
\newblock \emph{\bibinfo{journal}{Phys. Rev. Lett.}}
  \textbf{\bibinfo{volume}{119}}, \bibinfo{pages}{227401}
  (\bibinfo{year}{2017}).

\bibitem{hu2018Topological}
\bibinfo{author}{Hu, H.}, \bibinfo{author}{Hou, J.}, \bibinfo{author}{Zhang,
  F.} \& \bibinfo{author}{Zhang, C.}
\newblock \bibinfo{title}{Topological {{Triply Degenerate Points Induced}} by
  {{Spin}}-{{Tensor}}-{{Momentum Couplings}}}.
\newblock \emph{\bibinfo{journal}{Phys. Rev. Lett.}}
  \textbf{\bibinfo{volume}{120}}, \bibinfo{pages}{240401}
  (\bibinfo{year}{2018}).

\bibitem{yang2019Topological}
\bibinfo{author}{Yang, Y.} \emph{et~al.}
\newblock \bibinfo{title}{Topological triply degenerate point with double
  {{Fermi}} arcs}.
\newblock \emph{\bibinfo{journal}{Nat. Phys.}} \textbf{\bibinfo{volume}{15}},
  \bibinfo{pages}{645--649} (\bibinfo{year}{2019}).

\bibitem{zhang2018DoubleWeyl}
\bibinfo{author}{Zhang, T.} \emph{et~al.}
\newblock \bibinfo{title}{Double-{{Weyl Phonons}} in {{Transition}}-{{Metal
  Monosilicides}}}.
\newblock \emph{\bibinfo{journal}{Phys. Rev. Lett.}}
  \textbf{\bibinfo{volume}{120}}, \bibinfo{pages}{016401}
  (\bibinfo{year}{2018}).

\bibitem{zhu2016Triple}
\bibinfo{author}{Zhu, Z.}, \bibinfo{author}{Winkler, G.~W.},
  \bibinfo{author}{Wu, Q.}, \bibinfo{author}{Li, J.} \&
  \bibinfo{author}{Soluyanov, A.~A.}
\newblock \bibinfo{title}{Triple {{Point Topological Metals}}}.
\newblock \emph{\bibinfo{journal}{Phys. Rev. X}} \textbf{\bibinfo{volume}{6}},
  \bibinfo{pages}{031003} (\bibinfo{year}{2016}).

\bibitem{chang2017Nexus}
\bibinfo{author}{Chang, G.} \emph{et~al.}
\newblock \bibinfo{title}{Nexus fermions in topological symmorphic crystalline
  metals}.
\newblock \emph{\bibinfo{journal}{Sci. Rep.}} \textbf{\bibinfo{volume}{7}},
  \bibinfo{pages}{1688} (\bibinfo{year}{2017}).

\bibitem{lv2017Observation}
\bibinfo{author}{Lv, B.~Q.} \emph{et~al.}
\newblock \bibinfo{title}{Observation of three-component fermions in the
  topological semimetal molybdenum phosphide}.
\newblock \emph{\bibinfo{journal}{Nature}} \textbf{\bibinfo{volume}{546}},
  \bibinfo{pages}{627--631} (\bibinfo{year}{2017}).

\bibitem{zhang2018Topological}
\bibinfo{author}{Zhang, J.} \emph{et~al.}
\newblock \bibinfo{title}{Topological band crossings in hexagonal materials}.
\newblock \emph{\bibinfo{journal}{Phys. Rev. Materials}}
  \textbf{\bibinfo{volume}{2}}, \bibinfo{pages}{074201} (\bibinfo{year}{2018}).

\bibitem{chan2019Symmetryenforced}
\bibinfo{author}{Chan, Y.-H.} \emph{et~al.}
\newblock \bibinfo{title}{Symmetry-enforced band crossings in trigonal
  materials: {{Accordion}} states and {{Weyl}} nodal lines}.
\newblock \emph{\bibinfo{journal}{Phys. Rev. Mater.}}
  \textbf{\bibinfo{volume}{3}}, \bibinfo{pages}{124204} (\bibinfo{year}{2019}).

\bibitem{watanabe2018Space}
\bibinfo{author}{Watanabe, H.} \& \bibinfo{author}{Lu, L.}
\newblock \bibinfo{title}{Space {{Group Theory}} of {{Photonic Bands}}}.
\newblock \emph{\bibinfo{journal}{Phys. Rev. Lett.}}
  \textbf{\bibinfo{volume}{121}}, \bibinfo{pages}{263903}
  (\bibinfo{year}{2018}).

\bibitem{cornwall1999Center}
\bibinfo{author}{Cornwall, J.~M.}
\newblock \bibinfo{title}{Center vortices, nexuses, and the
  {{Georgi}}-{{Glashow}} model}.
\newblock \emph{\bibinfo{journal}{Phys. Rev. D}} \textbf{\bibinfo{volume}{59}},
  \bibinfo{pages}{125015} (\bibinfo{year}{1999}).

\bibitem{volovik2000Monopoles}
\bibinfo{author}{Volovik, G.~E.}
\newblock \bibinfo{title}{Monopoles and fractional vortices in chiral
  superconductors}.
\newblock \emph{\bibinfo{journal}{Proc. Natl. Acad. Sci. U.S.A.}}
  \textbf{\bibinfo{volume}{97}}, \bibinfo{pages}{2431--2436}
  (\bibinfo{year}{2000}).

\bibitem{heikkila2015Nexus}
\bibinfo{author}{Heikkil{\"a}, T.~T.} \& \bibinfo{author}{Volovik, G.~E.}
\newblock \bibinfo{title}{Nexus and {{Dirac}} lines in topological materials}.
\newblock \emph{\bibinfo{journal}{New J. Phys.}} \textbf{\bibinfo{volume}{17}},
  \bibinfo{pages}{093019} (\bibinfo{year}{2015}).

\bibitem{diebel2016Conical}
\bibinfo{author}{Diebel, F.}, \bibinfo{author}{Leykam, D.},
  \bibinfo{author}{Kroesen, S.}, \bibinfo{author}{Denz, C.} \&
  \bibinfo{author}{Desyatnikov, A.~S.}
\newblock \bibinfo{title}{Conical {{Diffraction}} and {{Composite Lieb Bosons}}
  in {{Photonic Lattices}}}.
\newblock \emph{\bibinfo{journal}{Phys. Rev. Lett.}}
  \textbf{\bibinfo{volume}{116}}, \bibinfo{pages}{183902}
  (\bibinfo{year}{2016}).

\bibitem{kopsky2002international}
\bibinfo{author}{Kopsky, V.} \& \bibinfo{author}{Litvin, D.}
\newblock \emph{\bibinfo{title}{International Tables for
  {{Crystallography}},{{Volume}} {{E}}: {{Subperiodic}} Groups}}.
\newblock International Tables for Crystallography (\bibinfo{year}{2002}).

\bibitem{mostafazadeh2004Pseudounitary}
\bibinfo{author}{Mostafazadeh, A.}
\newblock \bibinfo{title}{Pseudounitary operators and pseudounitary quantum
  dynamics}.
\newblock \emph{\bibinfo{journal}{J. Math. Phys.}}
  \textbf{\bibinfo{volume}{45}}, \bibinfo{pages}{932--946}
  (\bibinfo{year}{2004}).

\bibitem{green2010Isolated}
\bibinfo{author}{Green, D.}, \bibinfo{author}{Santos, L.} \&
  \bibinfo{author}{Chamon, C.}
\newblock \bibinfo{title}{Isolated flat bands and spin-1 conical bands in
  two-dimensional lattices}.
\newblock \emph{\bibinfo{journal}{Phys. Rev. B}} \textbf{\bibinfo{volume}{82}},
  \bibinfo{pages}{075104} (\bibinfo{year}{2010}).

\bibitem{huang2011Dirac}
\bibinfo{author}{Huang, X.}, \bibinfo{author}{Lai, Y.}, \bibinfo{author}{Hang,
  Z.~H.}, \bibinfo{author}{Zheng, H.} \& \bibinfo{author}{Chan, C.~T.}
\newblock \bibinfo{title}{Dirac cones induced by accidental degeneracy in
  photonic crystals and zero-refractive-index materials}.
\newblock \emph{\bibinfo{journal}{Nat. Mater.}} \textbf{\bibinfo{volume}{10}},
  \bibinfo{pages}{582--586} (\bibinfo{year}{2011}).

\bibitem{berry2007Chapter}
\bibinfo{author}{Berry, M.~V.} \& \bibinfo{author}{Jeffrey, M.~R.}
\newblock \bibinfo{title}{Chapter 2 {{Conical}} diffraction: {{Hamilton}}'s
  diabolical point at the heart of crystal optics}.
\newblock In \bibinfo{editor}{Wolf, E.} (ed.) \emph{\bibinfo{booktitle}{Prog.
  Opt.}}, vol.~\bibinfo{volume}{50}, \bibinfo{pages}{13--50}
  (\bibinfo{year}{2007}).

\bibitem{peleg2007Conical}
\bibinfo{author}{Peleg, O.} \emph{et~al.}
\newblock \bibinfo{title}{Conical {{Diffraction}} and {{Gap Solitons}} in
  {{Honeycomb Photonic Lattices}}}.
\newblock \emph{\bibinfo{journal}{Phys. Rev. Lett.}}
  \textbf{\bibinfo{volume}{98}}, \bibinfo{pages}{103901}
  (\bibinfo{year}{2007}).

\bibitem{liu2015Gauge}
\bibinfo{author}{Liu, F.} \& \bibinfo{author}{Li, J.}
\newblock \bibinfo{title}{Gauge {{Field Optics}} with {{Anisotropic Media}}}.
\newblock \emph{\bibinfo{journal}{Phys. Rev. Lett.}}
  \textbf{\bibinfo{volume}{114}}, \bibinfo{pages}{103902}
  (\bibinfo{year}{2015}).

\bibitem{chen2019NonAbelian}
\bibinfo{author}{Chen, Y.} \emph{et~al.}
\newblock \bibinfo{title}{Non-{{Abelian}} gauge field optics}.
\newblock \emph{\bibinfo{journal}{Nat. Commun.}} \textbf{\bibinfo{volume}{10}},
  \bibinfo{pages}{3125} (\bibinfo{year}{2019}).

\bibitem{jisha2019SelfTrapping}
\bibinfo{author}{Jisha, C.~P.}, \bibinfo{author}{Alberucci, A.},
  \bibinfo{author}{Beeckman, J.} \& \bibinfo{author}{Nolte, S.}
\newblock \bibinfo{title}{Self-{{Trapping}} of {{Light Using}} the
  {{Pancharatnam}}-{{Berry Phase}}}.
\newblock \emph{\bibinfo{journal}{Phys. Rev. X}} \textbf{\bibinfo{volume}{9}},
  \bibinfo{pages}{021051} (\bibinfo{year}{2019}).

\bibitem{rechcinska2019Engineering}
\bibinfo{author}{Rechci{\'n}ska, K.} \emph{et~al.}
\newblock \bibinfo{title}{Engineering spin-orbit synthetic {{Hamiltonians}} in
  liquid-crystal optical cavities}.
\newblock \emph{\bibinfo{journal}{Science}} \textbf{\bibinfo{volume}{366}},
  \bibinfo{pages}{727--730} (\bibinfo{year}{2019}).

\end{thebibliography}


\begin{thebibliography}{1}
\expandafter\ifx\csname url\endcsname\relax
  \def\url#1{\texttt{#1}}\fi
\expandafter\ifx\csname urlprefix\endcsname\relax\def\urlprefix{URL }\fi
\providecommand{\bibinfo}[2]{#2}
\providecommand{\eprint}[2][]{\url{#2}}

\bibitem{mostafazadeh2004Pseudounitary}
\bibinfo{author}{Mostafazadeh, A.}
\newblock \bibinfo{title}{Pseudounitary operators and pseudounitary quantum
  dynamics}.
\newblock \emph{\bibinfo{journal}{J. Math. Phys.}}
  \textbf{\bibinfo{volume}{45}}, \bibinfo{pages}{932--946}
  (\bibinfo{year}{2004}).

\bibitem{binding1996Eigencurves}
\bibinfo{author}{Binding, P.} \& \bibinfo{author}{Volkmer, H.}
\newblock \bibinfo{title}{Eigencurves for {{Two}}-{{Parameter
  Sturm}}-{{Liouville Equations}}}.
\newblock \emph{\bibinfo{journal}{SIAM Rev.}} \textbf{\bibinfo{volume}{38}},
  \bibinfo{pages}{27--48} (\bibinfo{year}{1996}).

\bibitem{kutsenko2013Spectral}
\bibinfo{author}{Kutsenko, A.~A.}, \bibinfo{author}{Shuvalov, A.~L.},
  \bibinfo{author}{Poncelet, O.} \& \bibinfo{author}{Norris, A.~N.}
\newblock \bibinfo{title}{Spectral properties of a {{2D}} scalar wave equation
  with {{1D}} periodic coefficients: {{Application}} to shear horizontal
  elastic waves}.
\newblock \emph{\bibinfo{journal}{Math. Mech Solids}}
  \textbf{\bibinfo{volume}{18}}, \bibinfo{pages}{677--700}
  (\bibinfo{year}{2013}).

\bibitem{sakoda2012Double}
\bibinfo{author}{Sakoda, K.}
\newblock \bibinfo{title}{Double {{Dirac}} cones in triangular-lattice
  metamaterials}.
\newblock \emph{\bibinfo{journal}{Opt. Express}} \textbf{\bibinfo{volume}{20}},
  \bibinfo{pages}{9925--9939} (\bibinfo{year}{2012}).

\bibitem{hu2018Topological}
\bibinfo{author}{Hu, H.}, \bibinfo{author}{Hou, J.}, \bibinfo{author}{Zhang,
  F.} \& \bibinfo{author}{Zhang, C.}
\newblock \bibinfo{title}{Topological {{Triply Degenerate Points Induced}} by
  {{Spin}}-{{Tensor}}-{{Momentum Couplings}}}.
\newblock \emph{\bibinfo{journal}{Phys. Rev. Lett.}}
  \textbf{\bibinfo{volume}{120}}, \bibinfo{pages}{240401}
  (\bibinfo{year}{2018}).

\bibitem{toth2011Quadrupolar}
\bibinfo{author}{T{\'o}th, T.~A.}
\newblock \bibinfo{title}{Quadrupolar {{Ordering}} in {{Two}}-{{Dimensional
  Spin}}-{{One Systems}}} (\bibinfo{year}{2011}).
\newblock \bibinfo{note}{Library Catalog: infoscience.epfl.ch}.

\end{thebibliography}

\end{document}


\title{Supplementary information for ``Hidden-symmetry-enforced nexus points of nodal lines in layer-stacked dielectric photonic crystals''}





\maketitle

\tableofcontents

\newpage
\section{Calculating Band structure using transfer matrix approach}
In this section, we derive the analytical expressions of the band structures in the $k_y=0$ plane for the AB-layer-stacked PhC. For a given angular frequency $\omega$ and a wavevector $\bm{\kappa}=(\kappa_x,0,k_z)$ in the $k_y=0$ plane, there are two plane wave eigensolutions in the homogeneous anisotropic medium described by Eq.~(2) in the main text. The two plane wave solutions can be labeled by their $\hat{M}_y$ parities, and their dispersion relations are given by 
\begin{align}
    \hat{M}_y-\text{odd}:&\quad \kappa_x^2+k_z^2-\varepsilon_{yy}k_0^2=0 &\Rightarrow\quad &\kappa^{\mathrm{odd}}_{x\pm}=\pm\kappa_1=\pm\sqrt{\varepsilon_{yy}k_0^2-k_z^2}, &\\
    \hat{M}_y-\text{even}:&\quad
    \varepsilon_{xx} \kappa_x^2+2\varepsilon_{xz}\kappa_xk_z+\varepsilon_{zz} k_z^2-\varepsilon_1\varepsilon_3k_0^2=0
     &\Rightarrow\quad &\kappa^{\mathrm{even}}_{x\pm}=\pm\kappa_2-\frac{\varepsilon_{xz}}{\varepsilon_{xx}}k_z=\pm\frac{\sqrt{\varepsilon_1\varepsilon_3(\varepsilon_{xx}k_0^2-k_z^2)}}{\varepsilon_{xx}}-\frac{\varepsilon_{xz}}{\varepsilon_{xx}}k_z, &
    \end{align}
    where $k_0=\omega/c$, $\varepsilon_1\varepsilon_3\equiv\varepsilon_{xx}\varepsilon_{zz}-\varepsilon_{xz}^2$, and we note that $\varepsilon_{xz}=\pm g$ in layer A and layer B respectively. The corresponding eigenvectors are
    \begin{align}
        \hat{M}_y-\text{odd}:&\quad \phi^\mathrm{odd}_\pm=(E_y,H_x,H_z)^\intercal=\left(\sqrt{\varepsilon_{yy}}k_0,-k_z,\pm\kappa_1\right)^\intercal \exp\left[i(\pm\kappa_1x+k_zz)\right],\label{planewave odd}\\
        \hat{M}_y-\text{even}:&\quad \phi^\mathrm{even}_\pm=(E_x,E_z,H_y)^\intercal=\left(\pm\frac{\varepsilon_{xz}\kappa_2}{\varepsilon_1\varepsilon_3}+\frac{k_z}{\varepsilon_{xx}},\mp\frac{\varepsilon_{xx}\kappa_2}{\varepsilon_1\varepsilon_3},k_0\right)^\intercal \exp\left[i(\kappa^\mathrm{even}_{x\pm}x+k_zz)\right].\label{planewave even}
    \end{align} 
    The Bloch eigenfunctions with certain $\hat{M}_y$ parity in the first period of the PhC ($x\in [-L/2,L/2]$\,) can be expressed as the superpositions of the plane wave fields in Eqs.~\eqref{planewave odd} and \eqref{planewave even}:
    \begin{align}
        \hat{M}_y-\text{odd}:&\qquad \psi^\mathrm{odd}_{\alpha}=a_\alpha\,\phi^{\mathrm{odd}}_{+,\alpha}  + b_\alpha\,\phi^{\mathrm{odd}}_{-,\alpha} ,\\
        \hat{M}_y-\text{even}:&\qquad \psi^\mathrm{even}_{\alpha}= c_\alpha\,\phi^{\mathrm{even}}_{+,\alpha}  + d_\alpha\,\phi^{\mathrm{even}}_{-,\alpha} .
    \end{align}
    where $\alpha=A,B$ labels the fields in layer A ($x\in [-L/2,0]$\,) and layer B ($x\in [0,L/2]$\,) respectively. And according to the Bloch condition, the field in the $m^\mathrm{th}$ period is given by $\psi_\alpha(x+mL)=\psi_\alpha(x) e^{imL\,k_x}$ with $k_x$ denoting the $x$ component of the Bloch wavevector inside the \nth{1} Brillouin zone (BZ).
    
    From the continuity conditions of $(E_y,H_z)$ and $(E_z,H_y)$ at the intracell interface  $x=0$, we obtain 
    \begin{equation}
    \begin{pmatrix}
    1&1&0&0\\
    1&-1&0&0\\
    0&0&-1&1\\
    0&0&1&1\\
    \end{pmatrix}
    \begin{pmatrix}
    a_A\\b_A\\c_A\\d_A\\
    \end{pmatrix}=
    \begin{pmatrix}
    1&1&0&0\\
    1&-1&0&0\\
    0&0&-1&1\\
    0&0&1&1\\
    \end{pmatrix}
    \begin{pmatrix}a_B\\b_B\\c_B\\d_B\\
    \end{pmatrix},
    \end{equation}
    which shows that $(a_A,b_A,c_A,d_A)^\intercal=(a_B,b_B,c_B,d_B)^\intercal$.
    Similarly, from the continuity boundary conditions at the intercell interface $x=L/2$, we obtian
    \begin{equation}
    \begin{pmatrix}
    e^{i\kappa_{1}L/2}&e^{-i\kappa_{1}L/2}&0&0\\
    e^{i\kappa_{1}L/2}&-e^{-i\kappa_{1}L/2}&0&0\\
    0&0&-e^{i\tilde{\kappa}_{+}L/2}&e^{-i\tilde{\kappa}_{-}L/2}\\
    0&0&e^{i\tilde{\kappa}_{+}L/2}&e^{-i\tilde{\kappa}_{-}L/2}\\
    \end{pmatrix}
    \begin{pmatrix}
    a_B\\b_B\\c_B\\d_B\\
    \end{pmatrix}
    =e^{iL\,k_x}
    \begin{pmatrix}
    e^{-i\kappa_{1}L/2}&e^{i\kappa_{1}L/2}&0&0\\
    e^{-i\kappa_{1}L/2}&-e^{i\kappa_{1}L/2}&0&0\\
    0&0&-e^{-i\tilde{\kappa}_{-}L/2}&e^{i\tilde{\kappa}_{+}L/2}\\
    0&0&e^{-i\tilde{\kappa}_{-}L/2}&e^{i\tilde{\kappa}_{+}L/2}\\
    \end{pmatrix}\begin{pmatrix}
    a_A\\b_A\\c_A\\d_A\\
    \end{pmatrix},
    \end{equation}
    where $\tilde{k}_\pm=\kappa_2\pm g k_z /\varepsilon_{xx}$. As a result, the transformer matrices for $\hat{M}_y$-odd and even modes can be written as
    \begin{align}
    \hat{M}_y-\text{odd}:&\qquad\begin{pmatrix}
     e^{i\kappa_{1}L}&0\\
    0&e^{-i\kappa_{1}L}\\
    \end{pmatrix}\begin{pmatrix}
    a_A\\b_A
    \end{pmatrix}=e^{iL\,k_x}\begin{pmatrix}
    a_A\\b_A
    \end{pmatrix},\\
    \hat{M}_y-\text{even}:&\qquad\begin{pmatrix}
     e^{i\kappa_{2}L}&0\\
    0&e^{-i\kappa_{2}L}\\
    \end{pmatrix}\begin{pmatrix}
    c_A\\d_A\\
    \end{pmatrix}=e^{iL\,k_x}\begin{pmatrix}
    c_A\\d_A\\
    \end{pmatrix}.
    \end{align}
    Solving the two equations, we find that the Bloch wavevector in the $x$ direction takes the simple expression $k_x=(\pm\kappa_1\ \mathrm{mod}\ 2\pi)$ for odd modes and $k_x=(\pm\kappa_2\ \mathrm{mod}\ 2\pi)$ for even modes. Therefore, the dispersions of $\hat{M}_y$-odd and even bands in the $k_y=0$ plane read
    \begin{align}\label{bandodd}
    \hat{M}_y-\text{odd}:&\qquad({\omega}_m^\mathrm{odd})^2/c^2  =\frac{1}{\varepsilon_{yy}}\left[\left(k_x+m\frac{2\pi}{L}\right)^2+k_z^2\right], \\
    \hat{M}_y-\text{even}:&\qquad({\omega}_m^\mathrm{even})^2/c^2  =\frac{\varepsilon_{xx}}{\varepsilon_1\varepsilon_3}\left(k_x+m\frac{2\pi}{L}\right)^2+\frac{1}{\varepsilon_{xx}} k_z^2,\label{bandeven}
    \end{align}
    with $m\in\mathbb{Z}$ numbering the bands. And the corresponding normalized Bloch states are
    \begin{align}\label{eigenfieldodd}
    &\ \hat{M}_y-\text{odd}:\quad
    \psi_m^\mathrm{odd}  =\left(E_y,H_z,H_z\right)^\intercal=\frac{1}{\sqrt{2\varepsilon_{yy} L}}\left(1,-\frac{ck_{z}}{\omega^\mathrm{odd}_m},\frac{c(k_x+m\frac{2\pi}{L})}{\omega^\mathrm{odd}_m}\right)^\intercal \exp\left[i\big(k_x+m\frac{2\pi}{L})x+ k_{z}z\big)\right],\\
    &\begin{aligned}\label{eigenfieldeven}
    \hat{M}_y-\text{even}:\quad \psi_m^{even}  =\left(E_x,E_z,H_y\right)^\intercal =&\,\frac{c}{\sqrt{2L}\,\omega^\mathrm{even}_m}\left(\frac{\varepsilon_{xz}(k_x+m\frac{2\pi}{L})}{\varepsilon_1\varepsilon_3}+\frac{ k_{z}}{\varepsilon_{xx}},-\frac{\varepsilon_{xx}(k_x+m\frac{2\pi}{L})}{\varepsilon_1\varepsilon_3},\frac{\omega^\mathrm{even}_m}{c}\right)^\intercal\\
    &\cdot\exp\left[i \left((k_x+n\frac{2\pi}{L})-\frac{\varepsilon_{xz}}{\varepsilon_{xx}}k_{z} \right){x}+i k_{z}z\right],
    \end{aligned}
    \end{align}
    Eqs.~\eqref{bandodd} and \eqref{bandeven} exhibit that all $\hat{M}_y$-odd (even) bands have identical conical dispersions up to a translation along the $x$-axis in the extended BZ. Therefore, any pair of bands with opposite mirror parities will cross each other along a nodal ling in the $k_y=0$ plane. We note that the nodal line can be any kinds of conic sections, including ellipse, parabola, and hyperbola, depending on the parameters of the PhC. In particular, the two red nodal rings of our interest in the main text correspond to the intersections of the even band of $m^\mathrm{even}=0$ and the odd bands of $m^\mathrm{odd}=\pm1$, given by the following implicit equation: 
    \begin{equation}
        \frac{\varepsilon_{xx}-\varepsilon_{yy}}{\varepsilon_{xx}\varepsilon_{yy}}k_z^2=\frac{\varepsilon_{xx}}{\varepsilon_1\varepsilon_3}k_x^2-\frac{1}{\varepsilon_{yy}}\left(k_x^2\pm\frac{2\pi}{L}\right)^2.
    \end{equation}
    And from the Eqs.~\eqref{bandodd} and \eqref{bandeven} , we can also obtain the frequency and wavevectors of the pair of triply degenerate crossing points of the two nodal rings, \textit{i.e.} the triply degenerate nexus points: $\omega^\mathrm{NP}=\frac{2\pi c}{L\sqrt{\varepsilon_{yy}-\varepsilon_{xx}}}$ and $\mathbf{k}^\mathrm{NP_\pm}=(0,0,\pm\frac{2\pi}{L}\sqrt{\frac{\varepsilon_{xx}}{\varepsilon_{yy}-\varepsilon_{xx}}})$. 
    
    \section{Hidden symmetries in the $k_y=0$ plane}
    In this section, we discuss the hidden symmetry in the $k_y=0$ plane induced by the fractional periodicity of the components of constitutive tensors. In this subsystem, the Maxwell's equations in the layer-stacked dielectric PhC can be written as
    \begin{equation}\label{maxwell}
    \underbrace{
         \left(
        \begin{array}{c@{}|c@{}}
       \mbox{\large0} &    \begin{smallmatrix}
            0 & k_z & 0 \\[2pt]
            -k_z & 0 & -i\partial_x\\[2pt]
            0 & i\partial_x & 0 \rule[-1ex]{0pt}{2ex}
          \end{smallmatrix} \\\hline
           \begin{smallmatrix}\rule{0pt}{2ex}
             0 & -k_z & 0 \\[2pt]
            k_z & 0 & i\partial_x\\[2pt]
            0 & -i\partial_x & 0 \rule[-1ex]{0pt}{2ex}
          \end{smallmatrix}\quad   &  \mbox{\large0}
        \end{array} 
    \right)
    }_{\displaystyle\hat{\mathcal{N}}}
    \underbrace{\begin{pmatrix}
        \mathbf{E}\\ \mathbf{H}
        \end{pmatrix}}_{\displaystyle\Psi}
        =\omega 
    \underbrace{\begin{pmatrix}
        \varepsilon_0\tensor{\varepsilon}_r(x) & 0)\\
        0 & \mu_0\hat{I}_{3\times3}
        \end{pmatrix}}_{\displaystyle\hat{\mathcal{M}}(x)}
        \begin{pmatrix}
        \mathbf{E}\\ \mathbf{H}
        \end{pmatrix},
    \end{equation}
    where $\tensor{\varepsilon}_r$ takes the form of Eq.~(2) in the main text. And hereinafter, we adopt the natural units with $\varepsilon_0=\mu_0=1$ for convenience.  In general, the Maxwell's equations are invariant under a symmetry transformation $\widetilde{A}$, as long as 
    \begin{equation}\label{generic symmetry}
        \widetilde{A}\hat{\mathcal{N}}\widetilde{A}^{-1}=\hat{C}\hat{\mathcal{N}}\quad \text{and}\quad \widetilde{A}\hat{\mathcal{M}}\widetilde{A}^{-1}=\hat{C}\hat{\mathcal{M}},
    \end{equation}
    where $\hat{C}$ can be an arbitrary invertible operator. However, for space group symmetries of the structure, $\hat{C}$ is fixed as identity. If the constitutive tensor $\hat{\mathcal{M}}$ is invertible, which is always true for dielectric PhCs, Eq.~\eqref{generic symmetry} is equivalent to the invariance of the effective Hamiltonian $\hat{H}=\hat{\mathcal{M}}^{-1}\hat{\mathcal{N}}$, \textit{i.e.} $\widetilde{A}\hat{H}\widetilde{A}^{-1}=\hat{H}$. In the $k_y=0$ plane, the effective Hamiltonian reads
    \begin{equation}\label{hamiltonian}
        \hat{H}(k_y=0)=\hat{\mathcal{M}}^{-1}\hat{\mathcal{N}}=
        \left(\begin{array}{c|c}
        0 & \hat{H_1}\\\hline
        \hat{H}_2 & 0
        \end{array}\right)=
    \left(\begin{array}{c@{}|c@{}}
       \mbox{\large0} &    \begin{smallmatrix}
            0 & \frac{\varepsilon_{xz}}{\varepsilon_1\varepsilon_3}(-i\partial_x)+\frac{\varepsilon_{zz}}{\varepsilon_1\varepsilon_3}k_z & 0 \\[2pt]
            -\frac{k_z}{\varepsilon_{yy}} & 0 & \frac{1}{\varepsilon_{yy}}(-i\partial_x)\\[2pt]
            0 & \frac{\varepsilon_{xx}}{\varepsilon_1\varepsilon_3}(i\partial_x)-\frac{\varepsilon_{xz}}{\varepsilon_1\varepsilon_3}k_z & 0 \rule[-1ex]{0pt}{2ex}
          \end{smallmatrix} \\\hline
           \begin{smallmatrix}\rule{0pt}{2ex}
             0 & -k_z & 0 \\[2pt]
            k_z & 0 & i\partial_x\\[2pt]
            0 & -i\partial_x & 0 \rule[-1ex]{0pt}{2ex}
          \end{smallmatrix}\quad   &  \mbox{\large0}
        \end{array} 
    \right).
    \end{equation}
    In what follows, we show that the fractional periodicity of the elements of $\tensor{\varepsilon}_r$ can give rise to a hidden symmetry of $\hat{H}(k_y=0)$ beyond space groups.

    \subsection{Hidden symmetries in the $\hat{M}_y$-odd subspace}
    In the $\hat{M}_y$-odd subspace, since the electric field is polarized in the $y$ direction, the $\hat{M}_y$-odd  band structure is entirely determined by $\varepsilon_{yy}$. Especially, for the AB-layered PhC in Fig.~1 of the main text, $\varepsilon_{yy}$ is a global constant, and hence the band structure of the odd modes on the  $k_y=0$ plane is directly obtained by folding the light cone in a homogeneous medium with a constant permittivity. Consequently, all $\hat{M}_y$-odd bands are twofold degenerate along $\Gamma-Z$ except for the one connected to $\omega=|\mathbf{k}|=0$ point.   
    
    In fact, we can relax the condition of a constant $\varepsilon_{yy}$ so that the period of $\varepsilon_{yy}$ is a fraction $1/N$ of the primitive period $L$ of the whole PhC, namely $\varepsilon_{yy}(x+L/N)=\varepsilon_{yy}(x)$, the width of the genuine BZ of the $\hat{M}_y$-odd subspace (marked as BZ(odd)) should be $2N\pi/L$ which is $N$ times as large as the primitive BZ of the whole system (marked as BZ(whole)). Therefore, the finial band structure in BZ(whole) is obtained by translating the bands in BZ(odd) periodically with spacing $\Delta x=2\pi/L$. As a result, twofold degenerate Kramers-like NLs can appear along $\Gamma-Z$ as long as $N\geq3$. If we also require the PhC respects the space group $\mathbb{R}^2\rtimes \mathrm{Rod}(22)$, $N$ should be an even number and thus the minimal value of $N$ is 4. 
    We introduce the 1/4-period translation operator in the $\hat{M}_y$-odd subspace:
    \begin{equation}\label{translatonodd}
        \hat{T}_x^\mathrm{odd}(L/4)=\hat{T}_x(L/4)\hat{P}_-+\hat{P_+},
    \end{equation}
    where $\hat{T}_x(L/4)$ denotes the original translation operator, and $\hat{P}_\pm=\frac{1}{2}(\hat{I}\pm\hat{M}_y)$ is the projection operator onto $\hat{M}_y$-even(+)/odd(-) subspace. So the fractional periodicity of $\varepsilon_{yy}$ is equivalent to $\hat{T}_x^\mathrm{odd}(L/4)\hat{H}(k_y=0)\hat{T}_x^\mathrm{odd}(L/4)^{-1}=\hat{H}(k_y=0)$. At the same time, the $\hat{M}_y$-odd subsapce is also invariant under a generalized 1/4-period twofold screw operation about the $x$ axis:
    \begin{equation}\label{screwodd}
        \widetilde{S}^\mathrm{odd}_{L/4}=\left(\hat{C}_{2x}\hat{T}_x(L/4)\right)\hat{P}_-+\hat{P}_+,
    \end{equation}
    namely, the combination of a 1/4-period translation and a twofold rotation about the $x$ axis acting on the odd subspace. We will prove in the next section that the generalized 1/4-period twofold screw symmetry, together with the time reversal symmetry, guarantees the emergence of the Kramers-like NLs along $\Gamma-Z$.
    
    \subsection{Generalized fractional translation symmetry in the $\hat{M}_y$-even subspace}
    For the $\hat{M}_y$-even subspace, it is convenient to deal with the sub-Hamiltonian of Eq.~\eqref{hamiltonian} acting on $\psi^\mathrm{even}=(E_x,E_z,H_y)^\intercal$:
    \begin{equation}
        \hat{H}_\mathrm{even}=
        \begin{pmatrix}
        0 & 0 & \frac{\varepsilon_{xz}}{\varepsilon_1\varepsilon_3}(-i\partial_x)+\frac{\varepsilon_{zz}}{\varepsilon_1\varepsilon_3}k_z\\
        0 & 0 & \frac{\varepsilon_{xx}}{\varepsilon_1\varepsilon_3}(i\partial_x)-\frac{\varepsilon_{xz}}{\varepsilon_1\varepsilon_3}k_z \\
        k_z & i\partial_x & 0
        \end{pmatrix}.
    \end{equation}
    Though the period of the sub-Hamiltonian is the same as the full system, we will show that it can be converted into a new form with fractional period via a similarity transformation.  First, we change the eigen basis from $\psi^\mathrm{even}=(E_x,E_z,H_y)^\intercal$ to $\psi'^{\mathrm{even}}=\hat{G}(x)\psi^\mathrm{even}=(D_x,E_z,H_y)^\intercal$ with the transformation matrix:
    \begin{equation}
        \hat{G}(x)=\begin{pmatrix}
            \varepsilon_{xx}(x) & \varepsilon_{xz}(x) & 0\\
            0 & 1 & 0\\
            0 & 0 & 1
        \end{pmatrix},
    \end{equation}
    and the sub-Hamiltonian is transformed accordingly:
    \begin{equation}
        \hat{H}'_{\mathrm{even}}=\hat{G}\hat{H}_\mathrm{even}\hat{G}^{-1}=
        \begin{pmatrix}
        0 & 0 & k_z \\
        0 & 0 & -\frac{\varepsilon_{xx}}{\varepsilon_1\varepsilon_3}\left((-i\partial_x)+\frac{\varepsilon_{xz}}{\varepsilon_{xx}}k_z\right)\\
        \frac{k_z}{\varepsilon_{xx}} & -\left((-i\partial_x)+\frac{\varepsilon_{xz}}{\varepsilon_{xx}}k_z\right) & 0
        \end{pmatrix}.
    \end{equation}
    It shows that $-i\partial_x$ always appears together with an effective gauge potential $\mathcal{A}_x=k_z\frac{\varepsilon_{xz}}{\varepsilon_{xx}}$ in $\hat{H}'_\mathrm{even}$. Therefore, we can use a $U(1)$ gauge transformation $\hat{U}(x,k_z)=\exp\left[i\int_0^x\mathcal{A}_x(\xi)d\xi\right]=\exp\left[ik_z\int_0^x\frac{\varepsilon_{xz}(\xi)}{\varepsilon_{xx}(\xi)}d\xi\right]$ to remove the effective gauge potential in the Hamiltonian:
    \begin{equation}\label{hamiltonian even}
        \hat{H}''_\mathrm{even}=\hat{U}\hat{H}'_\mathrm{even}\hat{U}^\dagger=\hat{U}\hat{G}\hat{H}_\mathrm{even}\hat{G}^{-1}\hat{U}^\dagger
        =\begin{pmatrix}
        0 & 0 & k_z \\
        0 & 0 & \frac{\varepsilon_{xx}}{\varepsilon_1\varepsilon_3}(i\partial_x)\\
        \frac{k_z}{\varepsilon_{xx}} & i\partial_x & 0
        \end{pmatrix}.
    \end{equation}
    We note that the gauge transformation is $k_z$ dependent. Under the coordinate $z$ representation, it appears as a translation operator with $x$-dependent translation along the $z$ axis:
    \begin{equation}
        \hat{U}(x)=\int dk_z |k_z\rangle\langle k_z|\exp\left[ik_z\int_0^x\frac{\varepsilon_{xz}(\xi)}{\varepsilon_{xx}(\xi)}d\xi\right]=\exp\left[\left(\int_0^x\frac{\varepsilon_{xz}(\xi)}{\varepsilon_{xx}(\xi)}d\xi\right)\partial_z\right]=\hat{T}\left(-\left(\int_0^x\frac{\varepsilon_{xz}(\xi)}{\varepsilon_{xx}(\xi)}d\xi\right)\mathbf{e}_z\right).
    \end{equation}

    After the combined similarity transformation $(\hat{U}\hat{G})$, $\hat{H}''_\mathrm{even}$ is only explicitly dependent on $\varepsilon_{xx}$ and $\varepsilon_1\varepsilon_3=\varepsilon_{xx}\varepsilon_{zz}-(\varepsilon_{xz})^2$ as shown in Eq.~\eqref{hamiltonian even}. For the AB-layer-stacked PhC shown in Fig.~1 of the main text, both $\varepsilon_{xx}$ and $\varepsilon_1\varepsilon_3$ are constant. Therefore, the $\hat{M}_y$-even band structure are also obtained by folding the light cone of a homogeneous system, and all bands are twofold generate along $\Gamma-Z$  except for the lowest one attached at $\omega=|\mathbf{k}|=0$.
    
    Similar to the $\hat{M}_y$-odd case, if we relax the condition to that the common period of $\varepsilon_{xx}(x)$ and $\varepsilon_1(x)\varepsilon_3(x)$ is a fraction $1/N$ of the full period $L$ of the PhC with $N\geq3$, twofold degeneracies of even bands can still appear along the $\Gamma-Z$ line. 
    And without changing the space group of the system, the minimal value of $N$ is 4 and the corresponding components of the permittivity should satisfy 
    \begin{equation}\label{constraint}
        \varepsilon_{ii}(x+L/4)=\varepsilon_{ii}(x),\ (i=x,y,z),\quad \text{and}\quad \varepsilon_{xz}(x+L/4)^2=\varepsilon_{xz}(x)^2,\ \varepsilon_{xz}(x+L/2)=-\varepsilon_{xz}(x). 
    \end{equation}
    In the following, we will focus on this special case with minimized constraints on the PhC that are compatible with the space group $\mathbb{R}^2\rtimes\mathrm{Rod}(22)$ and support the Kramers-like NLs along $\Gamma-Z$.
    
    The $L/4$ periodicity of $\hat{H}''_\mathrm{even}$ can be regarded as a generalized fractional translation symmetry operating on the $\hat{M}_y$-even subspace for the original Hamiltonian~\eqref{hamiltonian}, $\widetilde{T}^\mathrm{even}_x(L/4)\hat{H}(k_y=0)\widetilde{T}^\mathrm{even}_x(L/4)^{-1}=\hat{H}(k_y=0)$, with
    \begin{equation}\label{translationeven}
        \widetilde{T}_x^\mathrm{even}(L/4)=\left(\hat{G}^{-1}\hat{U}^\dagger\hat{T}_x(L/4)\hat{U}\hat{G}\right)\hat{P}_++\hat{P}_-,
    \end{equation}
    where $\hat{G}$ is reformulated as $\hat{G}=\hat{I}_{6\times6}+(\varepsilon_{xx}(x)-1)\hat{\mathbf{e}}_1\hat{\mathbf{e}}_1+\varepsilon_{xz}(x)\hat{\mathbf{e}}_1\hat{\mathbf{e}}_3$ for the 6-dimensional eigenvector $\Psi=(\mathbf{E},\mathbf{H})^\intercal$. Moreover, since $\hat{H}''_{\mathrm{even}}$ also respects the twofold rotation symmetry $\hat{C}_{2x}$ about the $x$ axis, we can further introduce the generalized 1/4-period twofold screw operator:
    \begin{equation}\label{screweven}
        \widetilde{S}_{L/4}^\mathrm{even}=\left(\hat{G}^{-1}\hat{U}^\dagger\left(\hat{C}_{2x}\hat{T}_x(L/4)\right)\hat{U}\hat{G}\right)\hat{P}_++\hat{P}_-,
    \end{equation}
    and the Hamiltonian~\eqref{hamiltonian} is invariant under $\hat{S}_{L/4}^\mathrm{even}$.

    \subsection{Generalized 1/4-period translation and 1/4-period twofold screw symmetries in the $k_y=0$ plane}
    Combining Eqs.~\eqref{translatonodd},\eqref{translationeven} and Eqs.~\eqref{screwodd},\eqref{screweven} respectively, we obtain the generalized 1/4-period translation operator in the whole $k_y=0$ plane:
    \begin{equation}
        \begin{split}
          \widetilde{T}_x(L/4)&=\widetilde{T}^\mathrm{even}_x(L/4)\widetilde{T}^\mathrm{odd}_x(L/4)
        =\hat{T}_x(L/4)\hat{P_-}+\left(\hat{G}^{-1}\hat{U}^\dagger\hat{T}_x(L/N)\hat{U}\hat{G}\right)\hat{P}_+\\
        &=\widetilde{U}^{-1}\hat{T}_x(L/4)\widetilde{U}
        =\left(\hat{P}_-+\hat{G}^{-1}\hat{U}^\dagger\hat{P}_+\right)\hat{T}_x(L/4)\left(\hat{P}_-+\hat{U}\hat{G}\hat{P}_+\right),
        \end{split}
    \end{equation}
    and the generalized 1/4-period twofold screw operator in the $k_y=0$ plane:
    \begin{equation}
        \widetilde{S}_{L/4}=\widetilde{S}_{L/4}^\mathrm{even}\widetilde{S}_{L/4}^\mathrm{odd}
        =\left(\hat{C}_{2x}\hat{T}_x(L/4)\right)\hat{P}_-+\left(\hat{G}^{-1}\hat{U}^\dagger\left(\hat{C}_{2x}\hat{T}_x(L/4)\right)\hat{U}\hat{G}\right)\hat{P}_+
        =\widetilde{U}^{-1}\left(\hat{C}_{2x}\hat{T}_x(L/4)\right)\widetilde{U},
    \end{equation}
    where 
    \begin{equation}
        \widetilde{U}(x,k_z)=\hat{P}_-+\hat{U}\hat{G}\hat{P}_+=
        \left(\begin{array}{c@{}|c@{}}
        \begin{smallmatrix}
        e^{i\varphi(x,k_z)}\varepsilon_{xx}(x) & 0 & e^{i\varphi(x,k_z)}\varepsilon_{xz}(x)\ \\[1ex]
        0 & 1 & 0 \\[1ex]
        0 & 0 & e^{i\varphi(x,k_z)}\rule[-1ex]{0pt}{2ex}
        \end{smallmatrix}& \mbox{\large0}\\\hline
        \mbox{\large0} &
        \begin{smallmatrix}
        \rule[1ex]{0pt}{2ex}1 & 0 & 0\\[1ex]
        0 & e^{i\varphi(x,k_z)} & 0\\[1ex]
        0 & 0 & 1
        \end{smallmatrix}
        \end{array}\right)
    \end{equation}
    with $\varphi(x,k_z)=k_z\int_0^x\frac{\varepsilon_{xz}(\xi)}{\varepsilon_{xx}(\xi)}d\xi$,
    and $\widetilde{U}^{-1}=\hat{P}_-+\hat{G}^{-1}\hat{U}^\dagger\hat{P}_+$. It can be directly checked that the effective Hamiltonian given in Eq.~\eqref{hamiltonian} is invariant under $\widetilde{T}_x(L/4)$ and $\widetilde{S}_{L/4}$, providing that the permittivity of the PhC satisfies Eq.~\eqref{constraint},
    \begin{gather}
        \widetilde{T}_x(L/4)\hat{H}(k_y=0)\widetilde{T}_x(L/4)^{-1}=\hat{H}(k_y=0),\\
        \widetilde{S}_{L/4}\hat{H}(k_y=0)\widetilde{S}_{L/4}^{-1}=\hat{H}(k_y=0).
    \end{gather}
    Consequently, we have demonstrated that the fractional periodicity of the elements of permittivity tensor engenders the hidden symmetries of the Maxwell's equations.
    And we have the relation between the generalized 1/4-period screw rotation and the generalized fractional translation: $\widetilde{S}_{L/4}^2=\widetilde{T}_x(L/4)^2=\widetilde{T}_x(L/2)$.
    The generalized 1/4-period translation operator satisfies
    \begin{equation}
        \widetilde{T}_x(L/4)^\dagger=\widetilde{U}^\dagger\hat{T}_x(-L/4)(\widetilde{U}^{-1})^\dagger
        =\widetilde{U}^\dagger\widetilde{U}\left(\widetilde{U}^{-1}\hat{T}_x(-L/4)\widetilde{U}\right)\widetilde{U}^{-1}(\widetilde{U}^{-1})^\dagger
        =\left(\widetilde{U}^\dagger\widetilde{U}\right)\widetilde{T}_x(L/4)^{-1}\left(\widetilde{U}^\dagger\widetilde{U}\right)^{-1},
    \end{equation}
    and the similar result is established for the generalized screw rotation:
    \begin{equation}
        \widetilde{S}_{L/4}^\dagger=\left(\widetilde{U}^\dagger\widetilde{U}\right)\widetilde{S}_{L/4}^{-1}\left(\widetilde{U}^\dagger\widetilde{U}\right)^{-1}.
    \end{equation}
    Since $\widetilde{U}^\dagger\widetilde{U}$ is hermitian and positive definite, both the two hidden symmetry operators are  \textbf{$(\widetilde{U}^\dagger\widetilde{U})$-pseudo-unitary} and have unimodular eigenvalues~\cite{mostafazadeh2004Pseudounitary}. In particular, since $\left(\widetilde{S}_{L/4}\right)^4=\hat{T}_x(L)$, the eigenvalues of $\widetilde{S}_{L/4}$ for a Bloch state $\Psi(k_x,0,0)$ on the $k_x$-axis must be a forth root of $e^{ik_xL}$, \textit{i.e.} 
    \begin{equation}\label{branch index}
        \widetilde{S}_{L/4}\Psi^{(s)}(k_x,0,0)=s\,e^{ik_xL/4}\,\Psi^{(s)}(k_x,0,0),    
    \end{equation}
    with $s=\pm1,\pm i$ denoting the \textbf{$\widetilde{S}_{L/4}$ branch index} of the Bloch state. We note that the $\widetilde{S}_{L/4}$ branch index is only well defined for the states on the $k_x$-axis. Nevertheless, since $\widetilde{S}_{L/4}^2=\widetilde{T}_x(L/2)$, we have
    \begin{gather}
        \widetilde{S}_{L/4}^2\Psi(k_x,0,k_z)=\widetilde{T}_x(L/2)\Psi(k_x,0,k_y)=\pm e^{ik_x L/2}\Psi(k_x,0,k_z),\label{S2parity}\\
        \widetilde{S}_{L/4}^2\Psi^{(s)}(k_x,0,0)=s^2e^{ik_x L/2}\Psi^{(s)}(k_x,0,0).
    \end{gather}
    Therefore, the $\widetilde{S}_{L/4}^2$-parity (equal to the square of the branch index $s^2=\pm1$) is well defined for the whole band on the $k_y=0$ plane, and is determined by the branch index of the states $\Psi^{(s)}(k_x,0,0)$ on that band. For convenience, we will assign a ``\textbf{pseudo branch index}'' for every Bloch state on the $k_y=0$ plane as $\Psi^{(s)}(k_x,0,k_z)$, while we remind that only its square $s^2$ denoting the $\widetilde{S}_{L/4}^2$-parity of the state is  meaningful in general (the sign of the pseudo branch index for a state with $k_z\neq0$ is indeterminate, as $s^2=(-s)^2$), but the sign of $s$ makes practical sense for the states on the $k_x$-axis.
    
    

    \section{Kramers degeneracies induced by the generalized 1/4-period twofold screw symmetry}
    Here, we express the time reversal operator in the coordinate-momentum mixed representation $(x,k_y,k_z)$: $\mathcal{T}=\hat{\tau}\mathcal{K}(k_y\to-k_y,k_z\to-k_z)$, where $\hat{\tau}=\mathrm{diag}(\hat{I}_{3\times3},-\hat{I}_{3\times3})$ and $\mathcal{K}$ denotes complex conjugate. It is easy to check that $\mathcal{T}$ commutes with $\widetilde{U}$ and $\widetilde{S}_{L/4}$: $[\mathcal{T},\widetilde{U}]=0$, $[\mathcal{T},\widetilde{S}_{L/4}]=0$. As introduced in the main text, the combination of the generalized 1/4-period screw rotation operator $\widetilde{S}_{L/4}$ and the time reversal  gives a  \textbf{$(\widetilde{U}^\dagger\widetilde{U})$-pseudo-antiunitary} symmetry operation for the layer-stacked PhC:
    \begin{equation}
         \hat{\Theta}_{L/4}=\mathcal{T}\,\widetilde{S}_{L/4},
    \end{equation}
    Here, $(\widetilde{U}^\dagger\widetilde{U})$-pseudo-antiunitarity means that
    \begin{equation}
        \langle\hat{\Theta}_{L/4}\psi|\widetilde{U}^\dagger\widetilde{U}|\hat{\Theta}_{L/4}\phi\rangle=\langle\psi|\widetilde{U}^\dagger\widetilde{U}|\phi\rangle^*,
    \end{equation}
    where $\psi$, $\phi$ represent two arbitrary states.
    
    $\hat{\Theta}_{L/4}$ operating on a Bloch state $\Psi^{(s)}(k_x,0,k_z)$ on $k_y=0$ plane yields a new Bloch state $\widetilde{\Psi}(-k_x,0,k_z)=\hat{\Theta}_{L/4}\Psi^{(s)}(k_x,0,k_z)$ of the same frequency at $(-k_x,0,k_z)$. 
    In addition, since
    \begin{equation}
    \begin{split}
            \widetilde{S}_{L/4}\widetilde{\Psi}(-k_x,0,0)&=\widetilde{S}_{L/4}\hat{\Theta}_{L/4}\Psi^{(s)}(k_x,0,0)
        =\hat{\Theta}_{L/4}\widetilde{S}_{L/4}\Psi^{(s)}(k_x,0,0)\\
        &=\widetilde{S}_{L/4}\left(s e^{ik_xL/4}\Psi^{(s)}(k_x,0,0)\right)
        =s^*e^{-ik_xL/4} \widetilde{\Psi}(-k_x,0,0),
    \end{split}
    \end{equation}
    the new Bloch state $\widetilde{\Psi}(-k_x,0,k_z)=\widetilde{\Psi}^{(s^*)}(-k_x,0,k_z)$ has pseudo branch index $s^*$ and has the same $\widetilde{S}^2_{L/4}$-parity $(s^*)^2=s^2$ as $\Psi^{(s)}(k_x,0,k_z)$.
    
    In the mean time, 
    \begin{gather}
         \hat{\Theta}_{L/4}^2\Psi^{(s)}(k_x,0,k_z)=\widetilde{S}_{L/4}^2\Psi^{(s)}(k_x,0,k_z)=s^2e^{ik_xL/2}\Psi^{(s)}(k_x,0,k_z).
    \end{gather}
    In particular, on the $\Gamma-Z$ line ($k_x=k_y=0$), 
    $
        \hat{\Theta}_{L/4}^2\Psi^{(s)}(0,0,k_z)=s^2\Psi^{(s)}(0,0,k_z).
    $
    Consequently, if the $\widetilde{S}_{L/4}^2$-parity of the state equals $s^2=-1$ (\textit{i.e.} pseudo branch index of $\Psi^{(s)}(0,k_z)$ is $s=\pm i$), we also have
    \begin{equation}
        \hat{\Theta}_{L/4}^2\Psi^{(\pm i)}(0,0,k_z)=-\Psi^{(\pm i)}(0,0,k_z),
    \end{equation} 
    so the pseudo-antiunitary operator $\hat{\Theta}_{L/4}$ serves as a pseudo-Fermionic time reversal symmetry for the Bloch states with $\widetilde{S}_{L/4}^2$-parity $s^2=-1$ on the $\Gamma-Z$ line. And according to the similar derivation of Kramers theorem, we have
    \begin{equation}
    \begin{split}
          \langle\widetilde{\Psi}^{(\mp i)}(0,0,k_z)|\widetilde{U}^\dagger\widetilde{U}|\Psi^{(\pm i)}(0,0,k_z)\rangle
        &=\langle\hat{\Theta}_{L/4}\Psi^{(\pm i)}|\widetilde{U}^\dagger\widetilde{U}|\Psi^{(\pm i)}\rangle
        =\langle\hat{\Theta}_{L/4}^2\Psi^{(\pm i)}|\widetilde{U}^\dagger\widetilde{U}|\hat{\Theta}_{L/4}\Psi^{(\pm i)}\rangle^*\\
        &=-\langle\Psi^{(\pm i)}|\widetilde{U}^\dagger\widetilde{U}|\hat{\Theta}_{L/4}\Psi^{(\pm i)}\rangle^*
        =-\langle\Psi^{(\pm i)}|\widetilde{U}^\dagger\widetilde{U}|\hat{\Theta}_{L/4}\Psi^{(\pm i)}\rangle^\dagger\\
        &=-\langle\widetilde{\Psi}^{(\mp i)}(0,0,k_z)|\widetilde{U}^\dagger\widetilde{U}|{\Psi}^{(\pm i)}(0,0,k_z)\rangle\\[5pt]
        \Rightarrow\quad & \quad \langle\widetilde{\Psi}^{(\mp i)}(0,0,k_z)|\widetilde{U}^\dagger\widetilde{U}|{\Psi}^{(\pm i)}(0,0,k_z)\rangle=0,
    \end{split}
    \end{equation}
    where the second equality of the first line is due to the pseudo-antiunitarity of $\hat{\Theta}_{L/4}$. Since $\widetilde{U}^\dagger\widetilde{U}$ is positive definite, $\langle\widetilde{\Psi}(0,k_z)^{(\mp i)}|\widetilde{U}^\dagger\widetilde{U}|{\Psi}^{(\pm i)}(0,0,k_z)\rangle=0$ indicates that \textbf{$\widetilde{\Psi}^{(\mp i)}(0,0,k_z)=\hat{\Theta}_{L/4}\Psi^{(\pm i)}(0,0,k_z)$ and $\Psi^{(\pm i)}(0,0,k_z)$ must be two distinct and degenerate Bloch states}.  As a result, we have proved that \textbf{a Bloch band with negative $\widetilde{S}_{L/4}^2$ parity must intersect with another band with negative $\widetilde{S}_{L/4}^2$ parity along $\Gamma-Z$, forming Kramers-like NLs. Moreover, the pair of intersecting bands have certain branch indices $s=i$ and $s=-i$ along $\Gamma-X$}.

    On the other hand, it is well known that the ordinary twofold screw symmetry $\hat{S}_{2x}$ together with the time reversal symmetry $\mathcal{T}$ protects all Bloch states on the $k_x=\pi/L$ plane are doubly degenerate. And the two degenerate states $\Psi(\pi/L,k_y,k_z)$ and $\widetilde{\Psi}(\pi/L,k_y,k_z)$ are correlated by the combined antiunitary operator $\hat{\Theta}_{L/2}=\mathcal{T}\hat{S}_{2x}$:
    \begin{equation}
        \widetilde{\Psi}(\pi/L,k_y,k_z)= \hat{\Theta}_{L/2}{\Psi}(\pi/L,k_y,k_z).
    \end{equation}
    In particular, at X point ($\mathbf{k}=(\pi/L,0,0)$), if ${\Psi}^{(s)}(\pi/L,0,0)$ has branch index $s$, we have
    \begin{equation}
    \begin{split}
        \widetilde{S}_{L/4}\widetilde{\Psi}(\pi/L,0,0)&=\widetilde{S}_{L/4}\hat{\Theta}_{L/2}{\Psi}^{(s)}(\pi/L,0,0)
        =\hat{\Theta}_{L/2}\widetilde{S}_{L/4}{\Psi}^{(s)}(\pi/L,0,0)\\
        &=\hat{\Theta}_{L/2}\left(se^{i\pi/4}\Psi^{(s)}(\pi/L,0,0)\right)
        =s^*e^{-i\pi/4}\hat{\Theta}_{L/2}\Psi^{(s)}(\pi/L,0,0)\\
        &=(-is^*)e^{i\pi/4}\widetilde{\Psi}(\pi/L,0,0),
    \end{split}
    \end{equation}
    where the second equality of the first line is due to the fact that $\widetilde{S}_{L/4}$ and $\hat{\Theta}_{L/2}$ are commutable. This result shows that $\widetilde{\Psi}(\pi/L,0,0)=\widetilde{\Psi}^{(-is^*)}(\pi/L,0,0)$ has branch index $-is^*$. In other words, \textbf{the pair of Bloch bands (along the $k_x$-axis) degenerate at $X$ must have branch indices of either $s=1,s=-i$ or $s=-1,s=i$}.

    \section{Determining the branch indices of the two bands connected to $|\mathbf{k}|=\omega=0$}
    In this section, we figure out the branch indices of the two bands attached at $|\mathbf{k}|=\omega=0$. The expansion of the Bloch states $\Psi(k_x,0,0)$ at $k_x=\omega=0$ reads
    \begin{equation}
    \Psi(k_x,0,0)=e^{i k_x x}u_{k_x}\left(x\right)=e^{i k_x x}\left[u_0\left(x\right)+\left.\frac{\partial u_{k_x}}{\partial k_x}\right|_{k_x=0}k_x+\mathcal{O}\left(k_x^2\right)\right],
    \end{equation}
    with $u_0=(\mathbf{e}^0,\mathbf{h}^0)^\intercal$.
    Substitution of the expansion into the Maxwell's equation~\eqref{maxwell} yields 
    \begin{equation}
    \underbrace{\hat{\mathcal{N}}\left(-i\partial_x\right)u_0}_{\displaystyle \text{0-order:}\ \mathcal{O}(1)}+\underbrace{\left[\hat{\mathcal{N}}\left(k_x\right)u_0-\omega \hat{\mathcal{M}}u_0+\hat{\mathcal{N}}\left(-i\partial_x\right)\left.\frac{\partial u_{k_x}}{\partial k_x}\right|_{k_x=0}k_x\right]}_{\displaystyle \text{1-order:}\ \mathcal{O}(k_x)}+\mathcal{O}\left(k_x^2\right)=0.
    \end{equation}
    \begin{align}
        \text{0-order:}&\quad\hat{\mathcal{N}}\left(-i\partial_x\right)u_0=0\quad\Rightarrow\quad (-i\partial_x\hat{\mathbf{x}})\times u_0=0\quad\Rightarrow\quad \partial_x(e_y^0,e_z^0,h_y^0,h_z^0)=0,\\
        \text{1-order:}&\quad  \underbrace{\hat{\mathcal{N}}\left(k_x\right)u_0}_{\displaystyle\perp\hat{\mathbf{x}}}-\omega \hat{\mathcal{M}}u_0+\underbrace{\hat{\mathcal{N}}\left(-i\partial_x\right)\left.\frac{\partial u_{k_x}}{\partial k_x}\right|_{k_x=0}k_x}_{\displaystyle\perp\hat{\mathbf{x}}}=0\quad\Rightarrow\quad \left(\hat{\mathcal{M}}u_0\right)\cdot\hat{\mathbf{x}}=(d_x^0,b_x^0)=0,
    \end{align}
    where $d_x^0,b_x^0$ denote the $x$ components of the 0-order $\mathbf{D}$ and $\mathbf{B}$ fields respectively.
    When $k_z=0$, $\hat{U}$ is reduced to the identity matrix, thus $\widetilde{U}\left(k_z=0\right)=\left(\hat{P}_{-}+\hat{G}\hat{P}_{+}\right)$ and $\widetilde{u}_0=\widetilde{U}\left(k_z=0\right)u_0=\left(0,e_y^0,e_z^0,0,h_y^0,h_z^0\right)$ is a transverse constant vector. So we obtian
    \begin{equation}
        \lim_{k_x\to0}\widetilde{S}_{L/4}\Psi(k_x,0,0)=\widetilde{S}_{L/4}u_0=\widetilde{U}^{-1}\hat{C}_{2x}\hat{T}_x(L/4)\widetilde{u}_0=\widetilde{U}^{-1}(-\widetilde{u}_0)=-u_0.
    \end{equation}
    Therefore, \textbf{both the two bands stemming from $|\mathbf{k}|=\omega=0$ have the same branch index $s=-1$}.

    \section{Asymptotic dispersion of bands at infinity}
    For a general permitivitiy given by Eq.~(1) of the main text,
    the wave equations of $\hat{M}_y$-odd and even modes on the $k_y=0$ plane are, respectively,
    \begin{align}
    \hat{M}_y\text{-odd:}&\qquad\left[-\varepsilon_{yy}^{-1/2}\frac{d^2}{dx^2}\varepsilon_{yy}^{-1/2}+k_z^2\varepsilon_{yy}^{-1}\right]\left(\sqrt{\varepsilon_{yy}}E_y\right)=\frac{\omega^2}{c^2}\left(\sqrt{\varepsilon_{yy}}E_y\right)\\    
    \label{generaleven}
    \hat{M}_y\text{-even:}&\qquad\left[-\frac{d}{dx}\frac{\varepsilon_{xx}}{\varepsilon_1\varepsilon_3}\frac{d}{dx}+k_z\left\{-i\frac{d}{dx},\frac{\varepsilon_{xz}}{\varepsilon_1\varepsilon_3}\right\}+k_z^2 \frac{\varepsilon_{zz}}{\varepsilon_1\varepsilon_3}\right]H_y=\frac{\omega^2}{c^2} H_y
    \end{align}
    where $\{\cdot,\cdot\}$ denotes the anticommutator. By introducing a new function $\widetilde{H}_y\left(x\right)=\hat{U}H_y=\exp{\left(ik_z\int_0^x\frac{\varepsilon_{xz}(\xi)}{\varepsilon_{xx}(\xi)}d\xi\right)}H_y$, Eq.\eqref{generaleven} can be transformed into a standard Sturm-Liouville equation:
    \begin{equation}\label{sturm}
    \left[-\frac{d}{dx}\frac{\varepsilon_{xx}}{ \varepsilon_1\varepsilon_3}\frac{d}{dx}+k_z^2\varepsilon_{xx}^{-1}\right]\widetilde{H}_y=\frac{\omega^2}{c^2} \widetilde{H}_y.
    \end{equation}
    Therefore, the eigen-equations for $\hat{M}_y$ even and odd bands can be uniformly expressed as follows:
    \begin{equation}\label{sturm2}
    \left[\hat{K}+k_z^2V\left(x\right)\right]\psi=\frac{\omega^2}{c^2}\psi,
    \end{equation}
    with Bloch boundary condition $\psi\left(0\right)=e^{i k_x L}\psi\left(L\right)$, $\frac{d}{dx}\psi\left(0\right)=e^{i k_x L}\frac{d}{dx}\psi\left(L\right)$,
    where $\hat{K}$ is a positive semidefinite operator arising from the fact that the permittivity tensor of dielectrics is positive definite, and $V\left(x\right)=V\left(x+L\right)$ is a positive piecewise smooth function (so it has infimum $\inf\left(V\right)=V_\mathrm{min}$).
    
    \begin{theorem}
     All bands $\omega_n\left(k_z\right)$ of Eq.\eqref{sturm2} tend to a unique asymptotic linear dispersion, as $k_z \to \infty$, and the asymptotic slope is determined by the infimum of $V\left(x\right)$: 
    \begin{equation}
    \lim_{k_z \to \infty}\frac{\omega_n}{k_z}=\lim_{k_z \to \infty}\frac{d \omega_n}{d k_z}=c\sqrt{V_\mathrm{min}}.    
    \end{equation}
    \end{theorem}
    
    \begin{proof}
    1) Lower bound: left product of $\psi_n\left(x\right)$ to Eq.\eqref{sturm2} yields $\omega_n^2/c^2=\langle\psi_n|\hat{K}|\psi_n\rangle+k_z^2\langle\psi_n|V|\psi_n\rangle \ge k_z^2 V_{\min}$, since $\hat{K}$ is positive semidefinite. Hence, we have
    $\frac{\omega_n}{k_z} \ge c\sqrt{V_{\min}}$.
    
    2) Upper bound: for a fixed $k_z$, all the eigenstates $\{\psi_i\}$ of Eq.\eqref{sturm2} from a compete basis of the function space $F$ satisfying the boundary conditions. And we consider the subspace $\mathrm{span}\{\psi_1,\cdots, \psi_{n-1} \}$ spanned by the eigensates corresponding to the first $n-1$ eigenvalues $\omega_1^2\le \cdots \le \omega_{n-1}^2 \le \omega_n^2 \le \cdots$, and its complement space is $F\backslash \text{span}\{\psi_1,\cdots \psi_{n-1}\}=\text{span}\{\psi_n, \psi_{n+1},\cdots\}$.
    
    On the other hand, $\forall \epsilon >0$, we can find $n$-D function space $F_n$ satisfying the boundary conditions, such that $\frac{\left<u|V|u\right>}{\left<u|u\right>}\le V_{\min}+\epsilon$, $\forall u\left(x\right)\not\equiv0 \in F_n$. And since $\dim F_n=n>\dim \text{span}\{\psi_1,\cdots,\psi_{n-1}\}$, there exists a state $v\left(x\right)\neq 0\in F_n\cap \mathrm{span}\{\psi_n,\psi_{n+1},\cdots\}$, so we can expand $v\left(x\right)$ using the eigenstates $\psi_i$ ($i \ge n$): $v\left(x\right)=\sum_{i=n}^\infty \alpha_i \psi_i \left(x\right)$, and accordingly
    \begin{displaymath}
    \frac{\omega_n^2}{c^2}\le \frac{\sum_{i=n}^{\infty}\left|\alpha_i\right|^2\omega_i^2/c^2}{\sum_{i=n}^\infty \left|\alpha_i\right|^2}=\frac{\langle v|\hat{K}+k_z^2V|v\rangle}{\langle v|v\rangle}\le K_{\max}+k_z^2\left(V_{\min}+\epsilon\right),
    \end{displaymath}
    where $K_{\max}=\max \left\{\frac{\langle v|\hat{K}|v\rangle}{\langle v|v\rangle}|u\left(x\right)\ne 0 \in F_n \right\}< \infty$ is the upper bound of $\hat{K}$ in the subspace $F_n$. Therefore, we have
    \begin{displaymath}
    \frac{\omega_n}{c\,k_z}\le \inf_{\epsilon>0} \sqrt{\frac{K_{\max}}{k_z^2}+V_{\min}+\epsilon}.
    \end{displaymath}
    As $k_z \to \infty$, both the lower and upper bounds of $\omega_n/k_z$ tend to $\inf_{\epsilon>0}\sqrt{V_{\min}+\epsilon}=\sqrt{V_{\min}}$, so we obtain the limit $\lim_{k_z \to \infty} \omega_n/k_z=c\sqrt{V_{\min}}$. And due to L'H\^{o}pital's rule, we also have $\lim_{k_z \to\infty} d\omega/dk_z=c\sqrt{V_{\min}}$.
    \end{proof}
    
    \textbf{Remark}: The theorem can also be proved using min-max theorem of eigenvalues \cite{binding1996Eigencurves}. Indeed, some
    procedures in the proof of min-max theorem have also been used in the above proof. In addition, Ref.
    \cite{kutsenko2013Spectral} offers a different proof.
    
    As a consequence, the asymptotic group velocities of even bands and order bands are, respectively,
    \begin{equation}
    \begin{split}
    \lim_{k_z \to \infty} \frac{d\omega^{\mathrm{odd}}_n}{d k_z}=\frac{c}{\sqrt{\varepsilon^\mathrm{max}_{yy}}}, \\
    \lim_{k_z \to \infty} \frac{d\omega^{\mathrm{even}}_n}{d k_z}=\frac{c}{\sqrt{\varepsilon^\mathrm{max}_{xx}}}.
    \end{split}
    \end{equation}
    
    This asymptotic group velocities are numerically verified, as shown in Fig.~2d of the main text. Therefore, along $\Gamma-Z$ direction, all bands tend to a linear dispersion as $k_z \to \infty$. All even and all odd bands have identical asymptotic group velocities respectively. For almost any dielectric photonic crystal respecting the symmetries, there exits a pair of triple points as the nexus of the nodal lines (intersection either of 1st even and 2nd,3rd odd bands or of 1st odd and 2nd,3rd even bands). The only exception (zero measure) occurs when the asymptotic group velocities of even and odd bands are identical, namely $\varepsilon^{\max}_{xx}=\varepsilon^{\max}_{yy}$.
    
    \begin{figure}[htb]
        \includegraphics[width=0.62\textwidth]{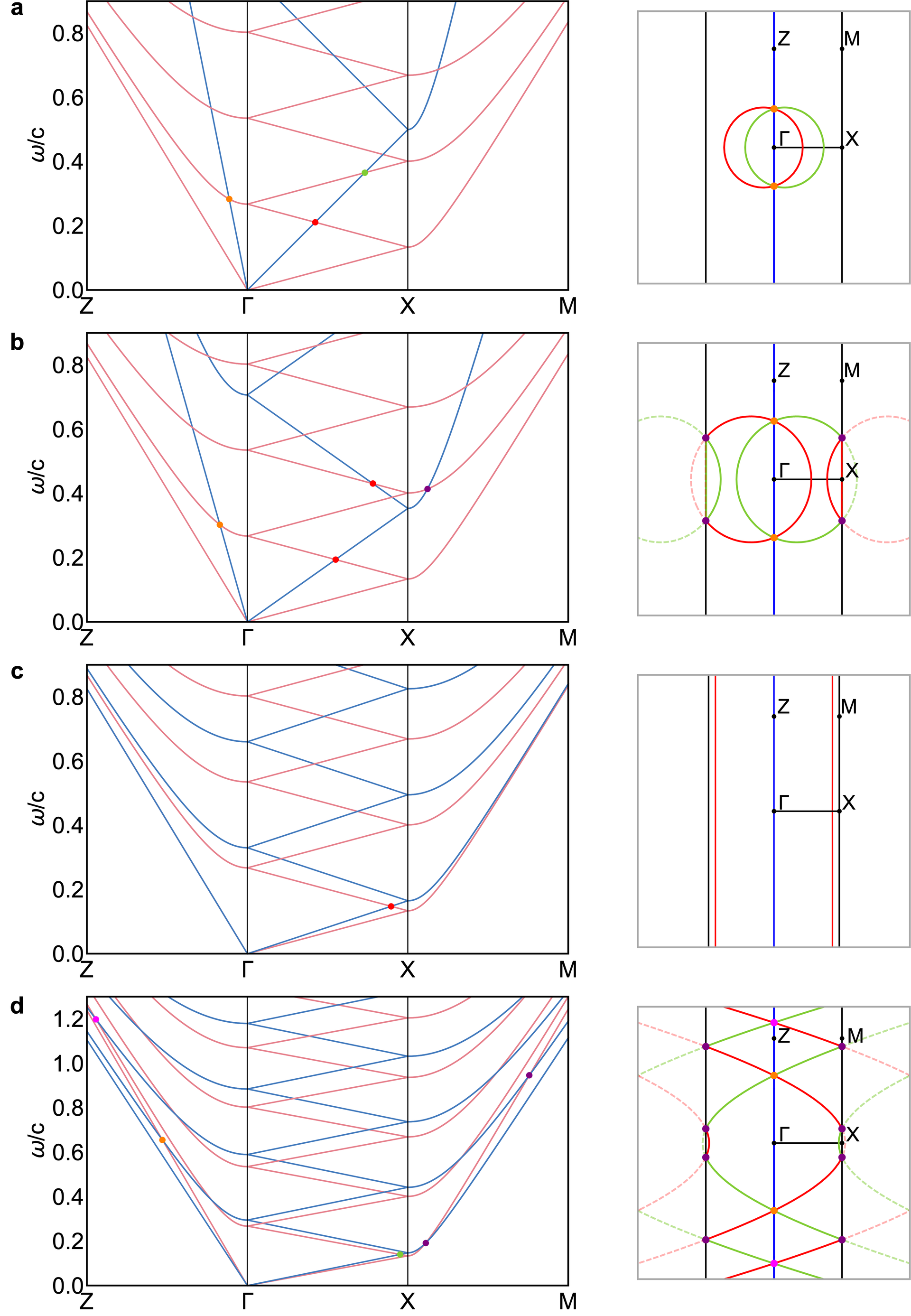}
        \caption{\label{parachange}  The band structures and nodal structures on the $k_y=0$ plane for the AB-layered PhCs with fixed parameters $\varepsilon_{xx}=9.17$, $\varepsilon_{zz}=14.83$, $|\varepsilon_{xz}|=2.83$, and  with different values of $\varepsilon_{yy}$ in each panel: \textbf{a}. $\varepsilon_{yy}=1$; \textbf{b}. $\varepsilon_{yy}=2$; \textbf{c}. $\varepsilon_{yy}=\varepsilon_{xx}=9.17$; \textbf{d}. $\varepsilon_{yy}=11.5$. The light blue and light red lines denote $\hat{M}_y$-odd and $\hat{M}_y$-even bands, respectively. Red and green dots and curves denote nodal lines protected by mirror and $\mathcal{PT}$ symmetries. Blue lines denote Kramers-like nodal lines. Orange dots denote triply degenerate nexus points. Purple and magenta dots denote fourfold degenerate nexus points.}
      \end{figure}

    \section{Robustness of nexus points against the variation  of  materials}
    In the main text, we have expounded that the special photonic band connectivity induced by the hidden symmetry guarantees that the triply degenerate NPs can almost always emerge on the 4 lowest bands along $\Gamma-Z$. 
    In Fig.~\ref{parachange}, we use the  biaxial dielectrics to construct the AB-layer-stacked PhCs and fixed the values of $\varepsilon_{xx}$, $\varepsilon_{zz}$, and $\varepsilon_{xz}$ for the PhCs in all the panels, while we change the value of $\varepsilon_{yy}$ in each panel to show how the nodal structure changes with the parameter of the PhC and to demonstrate the robustness of NPs. We note that there are infinite NLs in the band structures, while only those connecting with the lowest triple NPs are plotted.

    In Fig.~\ref{parachange}a, we let $\varepsilon_{yy}<\varepsilon_{xx}$. As such, the asymptotic group velocity of $\hat{M}_y$-odd bands (light blue) is larger than that of $\hat{M}_y$-even bands (light red) in the $\Gamma-Z$ direction. So the lowest triply degenerate NPs (orange dots) are formed by the \nth{1} odd and \nth{2}, \nth{3} even bands at the intersections of two nodal rings (red and green) and the lowest Kramers-like NL (blue).

    If we increase the value of $\varepsilon_{yy}$ while keeping it less than $\varepsilon_{xx}$, the two nodal rings will grow bigger and will eventually intersect with their transnational counterparts at the fourfold degenerate NPs (purple dots) on the two boundaries of the BZ as shown in Fig.~\ref{parachange}b. As a result, the nodal rings in the extended BZ connect together and form a nodal chain in the $k_y=0$ plane.
    
    As $\varepsilon_{yy}$ increases, the eccentricity of the two nodal rings grows accordingly, and the two triple NPs move outward along the $z$ axis from the origin. When $\varepsilon_{yy}$ reaches the critical value $\varepsilon_{yy}=\varepsilon_{xx}$, the asymptotic group velocities of even and odd bands are accidentally identical in the $z$ direction. Then, the two NPs move to the infinity and vanish, and the two nodal rings are reduced to two straight lines parallel to  $\Gamma-Z$, as shown in Fig.~\ref{parachange}c. As we have emphasized in the main text, the two triple NPs in the lowest 4 bands can only disappear in this accidental condition of $\varepsilon_{yy}=\varepsilon_{xx}$. However, if we arbitrarily select the parameters of the PhCs respecting the symmetries, the probability of encountering these exceptional cases is zero, since they are restricted to a subset of measure zero for all possible parameters.
    
    Once $\varepsilon_{yy}$ is larger than $\varepsilon_{xx}$, Fig.~\ref{parachange}d demonstrates that the two triple NPs will reappear immediately. However, in this case, the two NPs are on the band crossings of \nth{1} even and \nth{2},\nth{3} odd bands, as the asymptotic group velocity of even bands surpasses that of odd bands. At the same time, Fig.~\ref{parachange}d shows that the two previous ring shaped nodal lines convert to two hyperbolas. 
    And the nodal hyperbolas intersect with their translational counterparts infinite times on the $\Gamma-Z$ line as well as on the two BZ boundaries, forming fourfold degenerate NPs (magenta dots and purple dots).

    In Fig.~\ref{quadruple}, we plot the band structure near a typical fourfold degenerate NP on the BZ boundary ($X-M$). As shown in Fig.~\ref{quadruple}a,b, there are 4 nodal rings intersect on the NPs. In the section of $k_x=\pi/L$, there are two doubly degenerate bands intersecting at the NP as displayed in Fig.~\ref{quadruple}c, resulting from the combined symmetry of time reversal and twofold screw rotation $\hat{\Theta}_{L/2}=\mathcal{T}\hat{S}_{2x}$. Furthermore, Fig.~\ref{quadruple}d,e show that both the band structure in the $k_z=k^\mathrm{NP}_z$ section and the iso-frequency surfaces at $\omega^\mathrm{NP}$ disperse as 2D double Dirac cones~\cite{sakoda2012Double}.

    \begin{figure}[hb!]
    \includegraphics[width=0.75\textwidth]{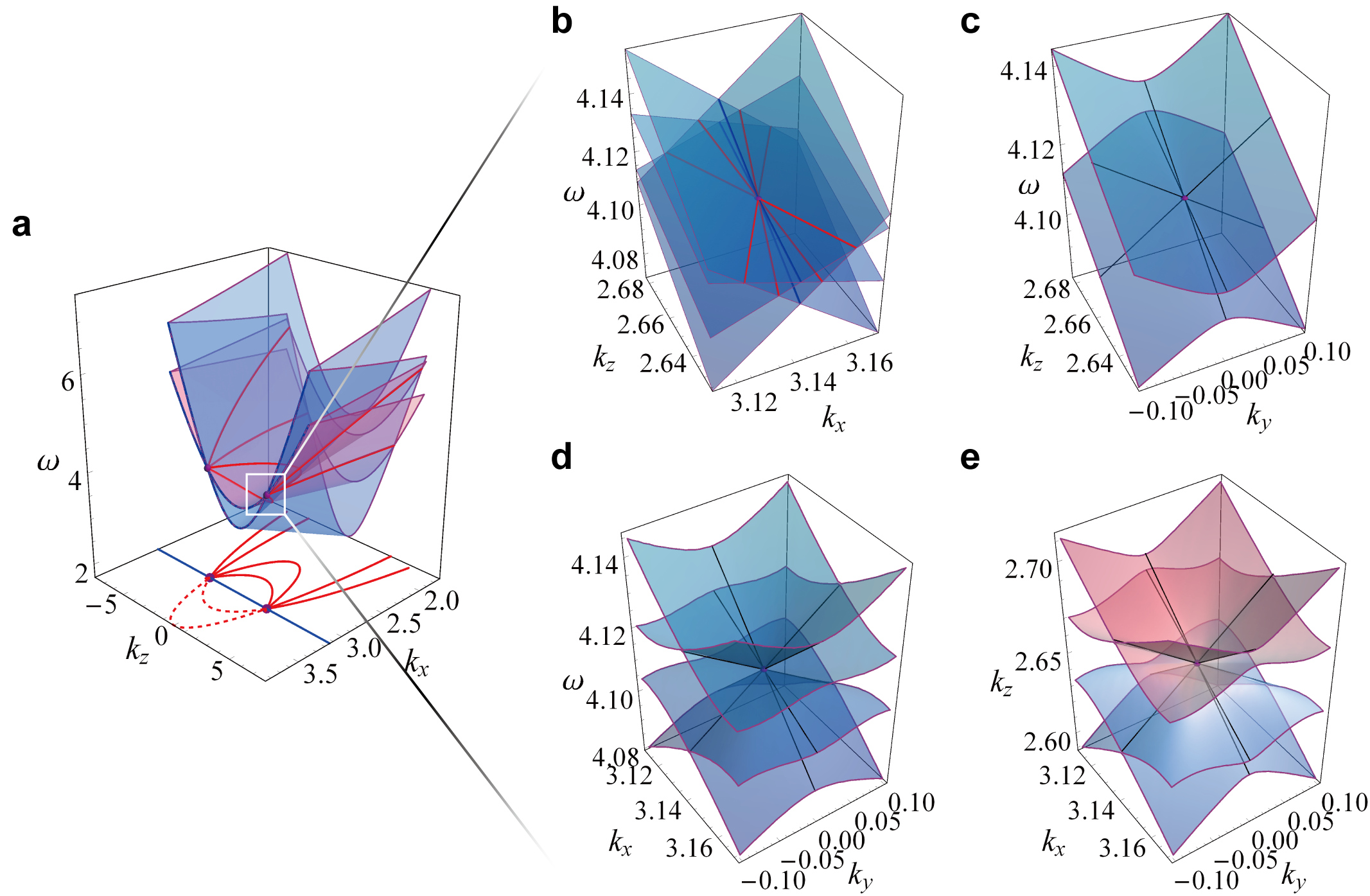}
    \caption{\label{quadruple}  Band structures near a fourfold degenerate nexus point in a AB-layer-stacked PhC with parameters $\varepsilon_1=\varepsilon_2=1$, $\varepsilon_3=12$ and $\theta=\pi/12$. \textbf{a}. Band structure near the nexus of 4 nodal lines on the $k_y=0$ plane. \textbf{b-d}, Zoomed in band structure near the NP in the sections of $k_y=0$, $k_x=\pi/L$, and $k_z=k^\mathrm{NP}_z$, respectively. \textbf{e}. Iso-frequency surfaces near the NP at the frequency of the NP.}
    \end{figure}

    \section{Derivation of $\mathbf{k}\cdot\mathbf{p}$ Hamiltonian for layer-stacked photonic crystals}
    Here we introduce a general $\mathbf{k}\cdot\mathbf{p}$ framework for a generic layer-stacked photonic crystal that is periodic in $x$ direction and homogeneous in the other two dimensions $\mathbf{r}_\perp=(y,z)$. In comparison with using wave equation (a 2-order PDE), it is more convenient to perform the derivation directly from the Maxwell's equations (1-order PDEs):
    \begin{equation}\label{maxwell2}
       \underbrace{\begin{pmatrix}
        0 & i\nabla\times\\
        -i\nabla\times & 0
        \end{pmatrix}}_{\displaystyle\hat{\mathcal{N}}}
        \underbrace{\begin{pmatrix}
        \mathbf{E}\\ \mathbf{H}
        \end{pmatrix}}_{\displaystyle\Psi}
        =\omega \underbrace{\begin{pmatrix}
        \tensor{\varepsilon}(x) & \tensor{\chi}(x)\\
        \tensor{\chi}(x)^\dagger & \tensor{\mu}(x)
        \end{pmatrix}}_{\displaystyle\hat{\mathcal{M}}}
        \begin{pmatrix}
        \mathbf{E}\\ \mathbf{H}
        \end{pmatrix},
    \end{equation}
    Expending of the Bloch state ${\Psi _{n,\mathbf{k}}}\left( \mathbf{r} \right) = e^{i\mathbf{k} \cdot \mathbf{r}} u_{n,\mathbf{k}}\left( x \right) = e^{i\mathbf{k}_\perp \cdot \mathbf{r}_\perp}\widetilde{\Psi}_{n,\mathbf{k}}\left( x \right)$ using the states at $\mathbf{k}_0$ (usually a degenerate point), we obtain 
    \begin{equation}
        \Psi_{n,\mathbf{k}} = \sum_m A_{nm}(\mathbf{k})e^{i(\mathbf{k} - \mathbf{k}_0)\cdot \mathbf{r}}\Psi _{m,\mathbf{k}_0}.
    \end{equation}
    Substitution of the expansion into Eq.~\eqref{maxwell2} yields
    \begin{equation}\label{kdotp2}
    \begin{split}
    0 &= \left( \hat{\mathcal{N}} - \omega_{n,\mathbf{k}}\hat{\mathcal{M}} \right)\Psi _{n\mathbf{k}} = \sum_m \left[ \hat{\mathcal{N}}e^{i(\mathbf{k} - \mathbf{k}_0)\cdot\mathbf{r}}\Psi _{m,\mathbf{k}_0} - \omega_{n,\mathbf{k}}\hat{\mathcal{M}}e^{i(\mathbf{k} - \mathbf{k}_0)\cdot\mathbf{r}}\Psi_{m,{\mathbf{k}_0}} \right]A_{nm} \\
     &= e^{i(\mathbf{k} - \mathbf{k}_0)\cdot\mathbf{r}}\sum_m \left[ 
     \begin{pmatrix}
        0 & (\mathbf{k}_0-\mathbf{k})\times\\
        -(\mathbf{k}_0-\mathbf{k})\times & 0
        \end{pmatrix}
     + \left(\omega_{m,\mathbf{k}_0} - \omega_{n,\mathbf{k}} \right)\hat{\mathcal{M}} \right]\Psi_{m,\mathbf{k}_0}A_{nm}. 
    \end{split}
    \end{equation}
    Since the Brillouin zone of a layer-stacked PhC is infinite in the transverse plane, the orthogonality of the Bloch states should be written as
    \begin{equation}
    \left\langle \Psi_{n,\mathbf{k}} \right|\hat{\mathcal{M}}\left| \Psi_{m,\mathbf{k}'} \right\rangle  = \frac{1}{( 2\pi)^2}\int_{ -\infty }^\infty  e^{i\left( \mathbf{k}'_\perp - \mathbf{k}_\perp \right) \cdot\mathbf{r}_\perp }d\mathbf{ r}_\perp   \int_{-L/2}^{L/2} dx\widetilde{\Psi}_{n,\mathbf{k}}^\dagger\hat{\mathcal{M}}\widetilde{\Psi}_{m,\mathbf{k}'}  = \delta( \mathbf{k}'_\perp - \mathbf{k}_\perp)\delta _{nm}\delta _{k_x,k'_x},
    \end{equation}
    where $\mathbf{k}_\bot$ and $\mathbf{r}_\bot$ are transverse wavevector and position vector respectively, $\delta( \mathbf{k}'_\perp - \mathbf{k}_\perp)$ represents Dirac delta function, while the other two delta functions are Kronecker delta. Performing inner product between $\Psi_{m',\mathbf{k}_0}$ and Eq.~\eqref{kdotp2}, we obtain the $\mathbf{k}\cdot\mathbf{p}$ eigen-equation of the \nth{1}-order
    \begin{equation}
        \sum_m \left[ \int_{-L/2}^{L/2} dx\Psi_{m',\mathbf{k}_0}^\dagger 
        \begin{pmatrix}
        0 & (\mathbf{k}_0-\mathbf{k})\times\\
        -(\mathbf{k}_0-\mathbf{k})\times & 0
        \end{pmatrix}
        \Psi_{m,\mathbf{k}_0} + \left( \omega_{m,\mathbf{k}_0} - \omega_{n,\mathbf{k}} \right)\delta_{m'm} \right] A_{nm} = 0,
    \end{equation}
    which can be rewritten as
    \begin{equation}
    \sum\limits_m {{H_{m'm}}} {A_{nm}} = \left( \omega_{n,\mathbf{k}} - \omega_{m',\mathbf{k}_0} \right)A_{nm'},
    \end{equation}
    with the element of the \nth{1}-order $\mathbf{k} \cdot \mathbf{p}$ Hamiltonian
    \begin{equation}\label{kpH}
     H_{m'm}=(\mathbf{k} - \mathbf{k}_0) \cdot \int_{-L/2}^{L/2} dx \left[ \mathbf{E}_{m,\mathbf{k}_0} \times \mathbf{H}_{m',\mathbf{k}_0}^* - \mathbf{H}_{m,\mathbf{k}_0} \times \mathbf{E}_{m',\mathbf{k}_0}^* \right]=(\mathbf{k} - \mathbf{k}_0) \cdot \mathbf{p}_{m'm}. 
    \end{equation}

    \subsection{$\mathbf{k}\cdot\mathbf{p}$ Hamiltonian near the nexus points}
    At the triply degenerate NPs with $\mathbf{k}^{\mathrm{NP}_\pm}=(0,0,\pm\frac{2\pi}{L}\sqrt{\frac{\varepsilon_{xx}}{\varepsilon_{yy}-\varepsilon_{xx}}})$ and $\omega^\mathbf{NP}=\frac{2c\pi}{L\sqrt{\varepsilon_{yy}-\varepsilon_{xx}}}$ for the AB-layer-stacked PhC in Fig.~1 of the main text, the three degenerate eigenstates can be obtained from Eqs.~\eqref{eigenfieldodd} and \eqref{eigenfieldeven}:
    \begin{align}\label{NPfield1}
    \Psi^\mathrm{odd}_{1}(\mathbf{k}^{\mathrm{NP}_\pm})&=(\mathbf{E}_1,\mathbf{H}_1)= \frac{1}{\sqrt{2 L}}\left(0,1,0,\mp\sqrt{\frac{\varepsilon_{xx}}{\varepsilon_{yy}}},0,\sqrt{\frac{\varepsilon_{yy}-\varepsilon_{xx}}{\varepsilon_{yy}}}\right)^\intercal \exp\left[i\left(\frac{2\pi}{L}x+ k^{\mathrm{NP}_\pm}_{z}z\right)\right],\\
    \Psi^\mathrm{even}_{0}(\mathbf{k}^{\mathrm{NP}_\pm})&=(\mathbf{E}_2,\mathbf{H}_2)=  \frac{1}{\sqrt{2L}}\left(\frac{ 1}{\varepsilon_{xx}},0,0,0,1,0\right)^\intercal \exp\left[ik^{\mathrm{NP}_\pm}_{z}\left( -\frac{\varepsilon_{xz}}{\varepsilon_{xx}}x+z\right)\right],\\
    \Psi^\mathrm{odd}_{-1}(\mathbf{k}^{\mathrm{NP}_\pm})&=(\mathbf{E}_3,\mathbf{H}_3)= \frac{1}{\sqrt{2 L}}\left(0,1,0,\mp\sqrt{\frac{\varepsilon_{xx}}{\varepsilon_{yy}}},0,-\sqrt{\frac{\varepsilon_{yy}-\varepsilon_{xx}}{\varepsilon_{yy}}}\right)^\intercal \exp\left[i\left(-\frac{2\pi}{L}x+ k^{\mathrm{NP}_\pm}_{z}z\right)\right].\label{NPfield3}
    \end{align}
    Substituting Eqs.~\eqref{NPfield1}-\eqref{NPfield3} into Eq.~\eqref{kpH}, we then obtain the \nth{1}-order $\mathbf{k}\cdot\mathbf{p}$ Hamiltonian at the NPs:
    \begin{equation}
    \hat{H}^{\pm\prime}_\mathrm{NP}=
    \begin{pmatrix}
        v_x\delta k_x\pm v^\mathrm{odd}_z\delta k_z & \displaystyle\frac{-i\,v^*_{y0}}{\sqrt{2}}\delta k_y & 0 \\
        \displaystyle\frac{i\,v_{y0}}{\sqrt{2}}\delta k_y & \pm v^\mathrm{even}_z\delta k_z & \displaystyle\frac{-i\,v_{y0}}{\sqrt{2}}\delta k_y\\
        0 & \displaystyle\frac{i\,v^*_{y0}}{\sqrt{2}}\delta k_y & -v_x\delta k_x\pm v^\mathrm{odd}_z\delta k_z
        \end{pmatrix},
    \end{equation}
    where  $\delta\mathbf{k}=(\delta k_x,\delta k_y,\delta k_z)=\mathbf{k}-\mathbf{k}^{\mathrm{NP}_\pm}$, and
    \begin{equation}
        {v}_x=\frac{\sqrt{\varepsilon_{yy}-\varepsilon_{xx}}}{\varepsilon_{yy}},\quad
        {v}_{y0}=\frac{\varepsilon_{xz}( \varepsilon _{yy} - \varepsilon_{xx} )}{\pi \sqrt{2\varepsilon_{yy}} ( {\varepsilon_{xz}^2} - \varepsilon _{xx}\varepsilon _{yy} + \varepsilon _{xx}^2)}\left( 1 + \exp\left({\textstyle \frac{\mp i\varepsilon_{xz}\pi}{\sqrt{\varepsilon _{xx}(\varepsilon_{yy}-\varepsilon_{xx})}}}\right) \right),\quad
        {v}^\mathrm{odd}_{z}=\frac{1}{\sqrt{\varepsilon_{xx}}},\quad
        {v}^\mathrm{even}_{z}=\frac{\sqrt{\varepsilon_{xx}}}{\varepsilon_{yy}}.
    \end{equation}
    Under the unitary transformation $\hat{V}=\mathrm{diag}(1,v^*_{y0}/|v_{y0}|,1)$, the Hamitonian of the NPs converts to
    \begin{equation}
    \begin{split}\label{Hnp}
       \hat{H}^\pm_\mathrm{NP}=\hat{V}\hat{H}^{\pm\prime}_\mathrm{NP}V^\dagger=&\begin{pmatrix}
        v_x\delta k_x\pm v^\mathrm{odd}_z\delta k_z & \displaystyle\frac{-i\,v_y}{\sqrt{2}}\delta k_y & 0 \\
        \displaystyle\frac{i\,v_y}{\sqrt{2}}\delta k_y & \pm v^\mathrm{even}_z\delta k_z & \displaystyle\frac{-i\,v_y}{\sqrt{2}}\delta k_y\\
        0 & \displaystyle\frac{i\,v_y}{\sqrt{2}}\delta k_y & -v_x\delta k_x\pm v^\mathrm{odd}_z\delta k_z
        \end{pmatrix}
    \end{split}
    \end{equation}
    with $v_y=|v_{y0}|$. In general, a $3\times3$ Hermitian matrix can be expanded by the 8 generators of $\mathsf{su}(3)$ Lie algebra. The 8 generators are usually selected as the 8 Gell-Mann matrices, or alternatively selected as the 3 spin-1 operators together with 5 spin-1 quadrupolar operators~\cite{hu2018Topological,toth2011Quadrupolar}. Here, in order to exhibit the relation between the nexus points and spin-1 physics, we adopt the latter choice, so Eq.~\eqref{Hnp} can be rewritten as
    \begin{equation}
            \hat{H}^\pm_\mathrm{NP}=v_x\hat{S}_z\delta k_x+v_y\hat{S}_y\delta k_y\pm\left[q_z\hat{Q}_{zz}+v_{z0}\hat{I}\right]\delta k_z,
    \end{equation}
    where $q_z=v^\mathrm{odd}_z-v^\mathrm{even}_z$, $v_{z0}=\frac{1}{3}(2v^\mathrm{odd}_z+v^\mathrm{even}_z)$, $\hat{S}_i$ ($i=x,y,z$) denote the 3 spin-1 operators, and $\hat{Q}_{zz}$ denotes one of the spin-1 quadrupolar operator~\cite{hu2018Topological,toth2011Quadrupolar}, which take the forms
    \begin{equation}
        \hat{S}_x=\frac{1}{\sqrt{2}}\begin{pmatrix}
    0&1&0\\
    1&0&1\\
    0&1&0\\
    \end{pmatrix},\quad
    \hat{S}_y=\frac{1}{\sqrt{2}}\begin{pmatrix}
    0&-i&0\\
    i&0&-i\\
    0&i&0\\
    \end{pmatrix},\quad
    \hat{S}_z=\begin{pmatrix}
    1&0&0\\
    0&0&0\\
    0&0&-1\\
    \end{pmatrix},\quad
    \hat{Q}_{zz}=(\hat{S}_z)^2-\frac{1}{3}\sum_{i}(\hat{S}_i)^2=
    \frac{1}{{3}}\begin{pmatrix}
    1&0&0\\
    0&-2&0\\
    0&0&1\\
    \end{pmatrix}.
    \end{equation}

    \section{Spin-1 conical diffraction}
    
    According to the $\mathbf{k}\cdot\mathbf{p}$ Hamiltonian near an NP, e.g. $\mathbf{k}^{\mathrm{NP}_+}$, the eigen-equation for states on the iso-frequency surface at the frequency $\omega^\mathrm{NP}$, \textit{i.e.} $\delta\omega=0$, reads
    \begin{equation}
    \begin{split}
       \underbrace{\begin{pmatrix}
        v_x\delta k_x & \displaystyle\frac{-i\,v_y}{\sqrt{2}}\delta k_y & 0 \\
        \displaystyle\frac{i\,v_y}{\sqrt{2}}\delta k_y & 0 & \displaystyle\frac{-i\,v_y}{\sqrt{2}}\delta k_y\\
        0 & \displaystyle\frac{i\,v_y}{\sqrt{2}}\delta k_y & -v_x\delta k_x
        \end{pmatrix}}_{\displaystyle \widetilde{H}(\delta\mathbf{k}_{xy})=v_x\hat{S}_z\delta k_x+{v}_y\hat{S}_y\delta k_y}\widetilde{\psi}=-\delta k_z
        \begin{pmatrix}
        v_z^\mathrm{odd} & 0 & 0\\
        0 & v_z^\mathrm{even} & 0\\
        0 & 0 & v_z^\mathrm{odd}
        \end{pmatrix}\widetilde{\psi}.
    \end{split}
    \end{equation}
    Using the transformation $\hat{R}=\mathrm{diag}(1,\sqrt{v_z^\mathrm{even}/v_z^\mathrm{odd}},1)$ and replacing $\delta k_z\to -i\partial_z$, we obtain the effective Schr\"odinger equation for the states on the iso-frequency surface:
    \begin{equation}
        i\,v^\mathrm{odd}_z \frac{\partial}{\partial z} \left|\psi\right>=\hat{H}(\delta\mathbf{k}_{xy})\left|\psi\right>,
    \end{equation}
    with the eigenstate $|\psi\rangle=\hat{R}\widetilde{\psi}$, and the effective anisotropic 2D spin-1 Hamiltonian:
    \begin{equation}\label{2D spin1 hamiltonian}
        \hat{H}(\delta\mathbf{k}_{xy})=\hat{R}^{-1}\widetilde{H}(\delta\mathbf{k}_{xy})\hat{R}^{-1}=
        \begin{pmatrix}
            v_x\delta k_x & \displaystyle\frac{-i\,\tilde{v}_y}{\sqrt{2}}\delta k_y & 0 \\
        \displaystyle\frac{i\,\tilde{v}_y}{\sqrt{2}}\delta k_y & 0 & \displaystyle\frac{-i\,\tilde{v}_y}{\sqrt{2}}\delta k_y\\
        0 & \displaystyle\frac{i\,\tilde{v}_y}{\sqrt{2}}\delta k_y & -v_x\delta k_x
        \end{pmatrix}
        =v_x\hat{S}_z\delta k_x+\tilde{v}_y\hat{S}_y\delta k_y,
    \end{equation}
    where $\tilde{v}_y=\sqrt{v_z^\mathrm{even}/v_z^\mathrm{odd}}v_y=\sqrt{\varepsilon_{xx}/\varepsilon_{yy}}v_y$. 
    In addition, we can introduce the spin-1 ladder operators in $\hat{S}_x$ representation:
    \begin{equation}
    \hat{S}_{\pm}=\hat{S}_z\mp i\hat{S}_y=\begin{pmatrix}
    1&\mp 1/\sqrt{2}&0\\
    \pm 1/\sqrt{2}&0&\mp 1/\sqrt{2}\\
    0&\pm 1/\sqrt{2}&1
    \end{pmatrix},    
    \end{equation}
    which raises and lowers a spin quantum number of the eigenstates of $\hat{S}_x$, respectively,
    \begin{equation}
        \hat{S}_\pm\left|s\right>=\sqrt{2-s(s\pm1)}\left|s\pm1\right>,
    \end{equation}
    where $\left|s\right>$ represent an eigenstate of $\hat{S}_x$  with spin quantum number (eigenvalue) $s$, \textit{i.e.} $\hat{S}_x|s\rangle=s|s\rangle\ (s\in\{-1,0,1\}$). As such, the effective spin-1 Hamiltonian~\eqref{2D spin1 hamiltonian} can be rewritten using the ladder operators:
    \begin{equation}
        \hat{H}(\delta\mathbf{k}_{xy}) = \frac{\delta\tilde{k}(\delta\mathbf{k}_{xy})}{2}\left( e^{i\phi(\delta\mathbf{k}_{xy}) }{\hat{S}_+ } + e^{-i\phi(\delta\mathbf{k}_{xy})}\hat{S}_- \right),
    \end{equation}
    with $\delta\tilde{k}(\delta\mathbf{k}_{xy})$ and $\phi(\delta\mathbf{k}_{xy})$ denoting the modulus and argument of $\delta\tilde{k}e^{i\phi}=v_x\delta k_x+i\tilde{v}_y\delta k_y$. Moreover, the evolution of a state along the $z$ axis can be explicitly expressed using the evolution operator:
    \begin{equation}
      \left|\psi(z)\right\rangle  = \exp\left[\frac{-iz}{v^\mathrm{odd}_z}\hat{H}(\delta\mathbf{k}_{xy}) \right]\left|\psi_0\right\rangle  = \left\{ I + \sum\limits_{n = 1}^\infty  \frac{1}{n!}\left[  \frac{-iz}{v^\mathrm{odd}_z}\frac{\delta\tilde{k}}{2}\left( e^{i\phi}\hat{S}_+  + e^{ - i\phi }\hat{S}_-  \right) \right]^n  \right\}\left|\psi_0 \right\rangle  
    \end{equation}
    where $|\psi_0\rangle$ represents the input state at $z=0$. 
    Using the ladder operators, we can obtain the following formulas:
    \begin{equation}
         \hat{h}^{n}|0\rangle=\left\{
         \begin{aligned}
         &\frac{2^n}{\sqrt{2}}\left(e^{i\phi}|1\rangle+e^{-i\phi}|-1\rangle\right),  &(n\in\text{odd})\\
         &2^{n} |0\rangle,  &(n\in\text{even}),
         \end{aligned}
         \right.
         \qquad\text{and}\qquad
         \hat{h}^n|\pm1\rangle=\sqrt{2}e^{\mp i\phi}\left(\hat{h}^{n-1}|0\rangle\right),
    \end{equation}
    with $\hat{h}=e^{i\phi}\hat{S}_++e^{-i\phi}\hat{S}_-$. Using these formulas, we can derive the final states evolving from different eigenstates of $\hat{S}_x$ as input:
    \begin{equation}
    \begin{split}
    \exp\left[\frac{-iz}{v^\mathrm{odd}_z}\hat{H}(\delta\mathbf{k}_{xy}) \right]\left| 1 \right> & =
    \frac{1}{2}\left[ {\cos\left({\textstyle\frac{\delta\tilde{k}\,z}{v^\mathrm{odd}_z}}\right) + 1} \right]\left| 1 \right\rangle  - \frac{i}{{\sqrt 2 }}\sin\left({\textstyle\frac{\delta\tilde{k}\,z}{v^\mathrm{odd}_z}}\right){e^{-i\phi }}\left| 0 \right\rangle  + \frac{1}{2}\left[ {\cos\left({\textstyle\frac{\delta\tilde{k}\,z}{v^\mathrm{odd}_z}}\right) - 1} \right]{e^{-2i\phi }}\left| { - 1} \right\rangle,\\ 
    \exp\left[\frac{-iz}{v^\mathrm{odd}_z}\hat{H}(\delta\mathbf{k}_{xy}) \right]\left| 0 \right\rangle & = \cos\left({\textstyle\frac{\delta\tilde{k}\,z}{v^\mathrm{odd}_z}}\right)\left| 0 \right\rangle  - \frac{i}{{\sqrt 2 }}\sin\left({\textstyle\frac{\delta\tilde{k}\,z}{v^\mathrm{odd}_z}}\right)\left( {{e^{i\phi }}\left| 1 \right\rangle  + {e^{ - i\phi }}\left| { - 1} \right\rangle } \right),\\
    \exp\left[\frac{-iz}{v^\mathrm{odd}_z}\hat{H}(\delta\mathbf{k}_{xy}) \right]\left| { - 1} \right\rangle  & = \frac{1}{2}\left[ {\cos\left({\textstyle\frac{\delta\tilde{k}\,z}{v^\mathrm{odd}_z}}\right) + 1} \right]\left| { - 1} \right\rangle  - \frac{i}{{\sqrt 2 }}\sin\left({\textstyle\frac{\delta\tilde{k}\,z}{v^\mathrm{odd}_z}}\right){e^{i\phi }}\left| 0 \right\rangle  + \frac{1}{2}\left[ {\cos\left({\textstyle\frac{\delta\tilde{k}\,z}{v^\mathrm{odd}_z}}\right) - 1} \right]{e^{2i\phi }}\left| 1 \right\rangle.
    \end{split}
    \end{equation}
    If we project the final state onto an spin eigenstate $\left|s_f\right>$ of $\hat{S}_x$, we obtain the compact form of the transition amplitude from an input spin state $|s_i\rangle$ to the output state $|s_f\rangle$:
    \begin{equation}
    \big\langle {{s_f}} \big|\exp\left[\frac{-iz}{v^\mathrm{odd}_z}\hat{H}(\delta\mathbf{k}_{xy})\right]\big|s_i\big\rangle  = \exp\left[i\left(s_f-s_i\right)\phi(\delta\mathbf{k}_{xy})\right]\left(\frac{1}{2}\right)^{\frac{s_f^2 + s_i^2}{2}}\left[ {{i^{\left| {{s_f} - {s_i}} \right|}}\cos \left( {\frac{\delta\tilde{k}(\delta\mathbf{k}_{xy})}{v^\mathrm{odd}_z}z + \frac{\pi}{2}\left| {{s_f} - {s_i}} \right|} \right) + {s_i}{s_f}} \right].
    \end{equation}
    The result reveals that the phase of the output field winds $l=(s_f-s_i)$ times around $\delta\mathbf{k}_{xy}=0$. Correspondingly, the phase on the ring of conical diffraction in the real space forms an optical vortex carrying the charge $l=(s_f-s_i)\in\{0,\pm1,\pm2\}$. 
    
    Indeed, the equality of the charge of the generated optical vortex and the difference of the spin quantum number of finial and initial states reflects the conservation of the generalized total angular momentum during the diffraction process. Since the effective 2D spin-1 Hamiltonian~\eqref{2D spin1 hamiltonian} has anisotropic Fermi velocity in the $xy$ plane, the total angular momentum is not a conserved quantity of the Hamiltonian. Nevertheless, we can rewrite the Hamiltonian as
    \begin{equation}
        \hat{H}(\delta\mathbf{k}_{xy})=v_x\hat{S}^z\delta k_x+\tilde{v}_y\hat{S}^y\delta k_y={e_a}^i\hat{S}^a\delta k_i,
        \qquad (a=z,y,\ i=x,y)
    \end{equation}
    with regarding the anisotropic Fermi velocity tensor as the tetard  $\big({e_a}^i\big)=\tensor{v}_f=\mathrm{diag}(v_x,\tilde{v}_y)$ of an anisotropic space. The corresponding metric tensor of the space is given by $\big(g^{ij}\big)=\big(\delta^{ab}{e_a}^i{e_b}^j\big)=\mathrm{diag}(v_x^2,\tilde{v}_y^2)$. Then we can define a generalized orbital angular momentum operator for the anisotropic space:
    \begin{equation}
        \widetilde{L}_z=\frac{1}{\sqrt{\det(g^{ij})}}\epsilon_{zij}r^i\delta k^j=\frac{1}{\sqrt{\det(g^{ij})}}\epsilon_{zij}r^i g^{jl}\delta k_l=\frac{\tilde{v}_y}{v_x}x\delta k_y-\frac{v_x}{\tilde{v}_y}y\delta k_x.
    \end{equation}
    where $\epsilon_{zij}$ denotes the antisymmetric symbol, $(r^i)=(x,y)$. 
    In terms of the coordinate transformation $\delta\tilde{k}e^{i\phi}=v_x\delta k_x+i\tilde{v}_y\delta k_y$ where $\tilde{k}$,$\phi$ can be viewed as the generalized polar coordinates in the momentum space, we have
    \begin{equation}
        \frac{\partial}{\partial\phi}=\frac{\partial\,\delta k_x}{\partial\phi}\frac{\partial}{\partial\delta k_x}+\frac{\partial\,\delta k_y}{\partial\phi}\frac{\partial}{\partial\delta k_y}
        =\left(\frac{-\tilde{v}_y}{v_x}\delta k_y\right)\frac{\partial}{\partial\delta k_x}
        +\left(\frac{v_x}{\tilde{v}_y}\delta k_x\right)\frac{\partial}{\partial\delta k_y}
        =i\frac{\tilde{v}_y}{v_x}\left(i\frac{\partial}{\partial\delta k_x}\right)\delta k_y
        -i\frac{v_x}{\tilde{v}_y}\left(i\frac{\partial}{\partial\delta k_y}\right)\delta k_x.
    \end{equation}
    Therefore, in the momentum representation, the generalized orbital angular momentum operator can be alternatively expressed as
    \begin{equation}
        \widetilde{L}_z=\frac{\tilde{v}_y}{v_x}x\delta k_y-\frac{v_x}{\tilde{v}_y}y\delta k_x=-i\frac{\partial}{\partial\phi}.
    \end{equation}
    Thus the eigenstates of the generalized orbital angular momentum operator are exactly given by $\psi_l=\exp\left[il\phi(\delta{\mathbf{k}}_{xy})\right]$ with $l$ denoting the orbital angular quantum number:
    \begin{equation}
         \widetilde{L}_z\exp\left[il\phi(\delta{\mathbf{k}}_{xy})\right]=l\exp\left[il\phi(\delta{\mathbf{k}}_{xy})\right].
    \end{equation}
    In addition, the generalized total angular momentum operator is given by
    \begin{equation}
        \widetilde{J}=\widetilde{L}_z-\hat{S}_x,
    \end{equation}
    where the negative sign in font of $\hat{S}_x$ is utilized for the sake of the consistent chirality of the spin frame $(\hat{S}_z,\hat{S}_y,-\hat{S}_x)$ and of the coordinate frame $(x,y,z)$. Then it can be directly verified that the generalized total angular momentum commutes with the Hamiltonian:
    \begin{equation}
        \big[\widetilde{J},\hat{H}(\delta\mathbf{k}_{xy})\big]=0.
    \end{equation}
    Therefore, the generalized total angular momentum quantum number $j=l-s$ is conserved during the wave propagation along the $z$ direction. For an incident state with $l_i=0$ orbital angular momentum and $s_i$ spin momentum, the output state satisfies $j=l_f-s_f=0-s_i$, thus the final orbital quantum number $l_f=s_f-s_i$ is determined by the difference of the finial and initial spin quantum numbers.

\bibliographystyle{naturemag}
\bibliography{references}